%% file: main.tex
\documentclass[10pt, conference, letterpaper]{IEEEtran}
\IEEEoverridecommandlockouts
\usepackage{cite}
\usepackage{amsmath,amssymb,amsfonts}
\usepackage{algorithm}
\usepackage{algpseudocode}
\usepackage{algpseudocodex}
\usepackage{enumitem}
\usepackage{graphicx}
\usepackage{textcomp}
\usepackage{xcolor}
\usepackage{amsthm} 
\usepackage{amsmath}
\usepackage{url}
\usepackage{subcaption}
\usepackage{flushend}
\usepackage{caption}
\usepackage{booktabs}
\usepackage{diagbox}

\def\BibTeX{{\rm B\kern-.05em{\sc i\kern-.025em b}\kern-.08em
    T\kern-.1667em\lower.7ex\hbox{E}\kern-.125emX}}

\newcommand*{\circled}[1]{\lower.7ex\hbox{\tikz\draw (0pt, 0pt)%
    circle (.5em) node {\makebox[1em][c]{\small #1}};}}
\robustify{\circled}

\newtheorem{theorem}{Theorem}
\newtheorem{lemma}{Lemma}
\newtheorem{definition}{Definition}

\def\energyReductionPercentageLow{11.36}
\def\energyReductionPercentageHigh{26.07}
\def\lifeConsumptionReductionPercentageLow{11.15}
\def\lifeConsumptionReductionPercentageHigh{26.75}
\def\delayReductionPercentageLow{8.37}
\def\delayReductionPercentageHigh{32.77}
\def\ourFrameworkName{SusCO}

\begin{document}

\title{Commercial Dishes Can Be My Ladder:\\ Sustainable and Collaborative Data Offloading \\in LEO Satellite Networks}

\author{\IEEEauthorblockN{Yi Ching Chou\IEEEauthorrefmark{1}\thanks{Y. C. Chou and L. Chen contributed equally to this research.}\thanks{This is an extended version of the paper accepted to INFOCOM 2025.}, Long Chen\IEEEauthorrefmark{1}, Hengzhi Wang\IEEEauthorrefmark{2}, Feng Wang\IEEEauthorrefmark{3},\\ Hao Fang\IEEEauthorrefmark{1}, Haoyuan Zhao\IEEEauthorrefmark{1},  Miao Zhang\IEEEauthorrefmark{1}, Xiaoyi Fan\IEEEauthorrefmark{4}, Jiangchuan Liu\IEEEauthorrefmark{1}}

\IEEEauthorblockA{\IEEEauthorrefmark{1}School of Computing Science, Simon Fraser University, Canada}

\IEEEauthorblockA{\IEEEauthorrefmark{2}College of Computer Science and Software Engineering, Shenzhen University, China}

\IEEEauthorblockA{\IEEEauthorrefmark{3}Department of Computer and Information Science, The University of Mississippi, USA}

\IEEEauthorblockA{\IEEEauthorrefmark{4}Jiangxing Intelligence Inc., Shenzhen, China}

Emails: ycchou@sfu.ca, longchen.cs@ieee.org, whz@szu.edu.cn, fwang@cs.olemiss.edu, \\ \{fanghaof, hza127, mza94\}@sfu.ca, xiaoyifan@jiangxingai.com, jcliu@sfu.ca

}

\maketitle

\begin{abstract}
Low Earth Orbit (LEO) satellite networks, characterized by their high data throughput and low latency, have gained significant interest from both industry and academia. Routing data efficiently within these networks is essential for maintaining a high quality of service. However, current routing strategies, such as bent-pipe and inter-satellite link (ISL) routing, have their unique challenges. The bent-pipe strategy requires a dense deployment of dedicated ground stations, while the ISL-based strategy can negatively impact satellite battery lifespan due to increased traffic load, leading to sustainability issues. 

In this paper, we propose \emph{sustainable collaborative offloading}, a framework that orchestrates groups of existing commercial resources like ground stations and 5G base stations for data offloading. This orchestration enhances total capacity, overcoming the limitations of a single resource. We propose the \emph{collaborator group set construction algorithm} to construct candidate groups and the \emph{collaborator selection and total payment algorithm} to select offloading targets and determine payments no less than the costs. Extensive real-world-based simulations show that our solution significantly improves energy consumption, satellite service life, and end-to-end latency.
\end{abstract}


\input{sections/introduction}

\input{sections/background_motivation}

\input{sections/system_model}
\input{sections/analysis}
\input{sections/evaluation}

\input{sections/related_work}
\input{sections/further_discussion}
\input{sections/conclusion}

\input{sections/ack}

\bibliographystyle{IEEEtran}
\bibliography{reference}

\end{document}

%% file: sections/introduction.tex
\section{Introduction}
\label{sec:introduction}

Low Earth Orbit (LEO) satellite networks have attracted great interest from industry and academia due to their high throughput, low latency, and relatively low launch cost characteristics due to technological advancement \cite{su_broadband_2019}. LEO satellite networks are also believed to be an important component in the 6G network \cite{xiao_leo_6g_2024} since they can offer larger service coverage at less cost compared to terrestrial networks, providing truly ubiquitous Internet access services on Earth \cite{Giuliari_backbones_space_2020}.

An efficient routing strategy is significant for delivering high-quality services to users. LEO satellite networks currently use two main data routing strategies to reach ground destinations \cite{ma_starlink_measure_2023}: 1) bent-pipe routing and 2) inter-satellite link (ISL)-based routing. Each strategy, however, has unique challenges. In the bent-pipe strategy, each LEO satellite relays data between locations within the same coverage area, which requires a dense deployment of ground stations \cite{singh_community_2021, vas_l2d2_2021}. Conversely, the ISL-based strategy routes data between satellites using ISLs, which increases the traffic load on LEO satellites and can negatively impact their battery lifespan, leading to congestion and sustainability issues \cite{chou_sustainable_2022}.

\begin{figure}[t]
\centerline{\includegraphics[width=0.49\textwidth]{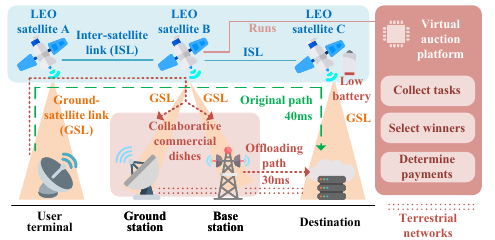}}
\caption{Sustainable and cost-effective \ourFrameworkName{} framework.}
\label{fig:vision}
\vspace{-6mm}
\end{figure}

In this paper, we argue that using solely these strategies is not the best for routing data in LEO satellite networks. We systematically study and answer the question: \emph{is there a cost-effective strategy that can utilize existing resources to reduce the need for new dedicated ground stations while jointly enhancing the sustainability and quality of service (QoS) in LEO satellite networks?}

One potential solution is to utilize existing commercial shared resources, such as commercial ground stations and 5G base stations, rather than constructing new dedicated ground stations. There exist numerous commercial shared resources, such as AWS Ground Station and Azure Orbital Ground Station \cite{aws_gs_sla_2022, azure_orbital_sla_2022}, and 5G base stations equipped with 5G New Radio \cite{3GPP_TR_38811}, that are capable of receiving data from satellites and offering \emph{Ground Station as a Service (GSaaS)}. As the satellite market expands, there will be more commercial shared resources offering GSaaS, similar to the Infrastructure as a Service model in cloud computing. For brevity, we refer to these shared resources as \emph{commercial dishes} in this paper. Utilizing these commercial dishes can significantly reduce the deployment costs for LEO satellite operators.

However, existing GSaaS offerings have high variation in terms of QoS, making it difficult to select the one that can maximize gains while minimizing costs. Furthermore, different commercial dishes have different Service Level Agreements with varying QoS and uptime guarantees, and a single commercial dish may not satisfy the requirements. Hence, selecting the group of commercial dishes under the satellite coverage that offers the most gains while maintaining high QoS is also a challenge.

To address these challenges, we propose a sustainable and cost-effective offloading framework named Sustainable Collaborative Offloading (\ourFrameworkName). In \ourFrameworkName, LEO satellite networks can offload data to a group of commercial dishes, mitigating congestion in satellite hops and prolonging LEO satellite service life without the need to construct new dedicated ground stations. The \ourFrameworkName{} orchestrates the collaboration among commercial dishes to enhance the total capacity \cite{singh_community_2021} such as higher bandwidth, lower end-to-end latency, and receiving more offloading data.

Fig. \ref{fig:vision} illustrates how \ourFrameworkName{} works. The user terminal generates and transmits data to the LEO satellite A, where the data can be routed on the original path (green dash line) but the LEO satellite C on the original path has a low battery level, where the energy consumption has a more negative impact on the battery lifespan for low battery levels \cite{yang_towards_2016}. On the other hand, the data can be offloaded from the LEO satellite B to a group of collaborative commercial dishes (red dash arrows) before the LEO satellite C and then further routed to the destination through terrestrial networks, which can reduce congestion\cite{li_leotp_2023} and significantly prolong the battery lifespan for the LEO satellite C. However, some questions must be carefully addressed before the LEO satellite networks can fully utilize the existing commercial dishes in \ourFrameworkName:

\begin{enumerate}
    \item How do we quantify the improvement in energy consumption, end-to-end latency, and satellite service life from offloading data to commercial dishes?

    \item How do we select a group of collaborative commercial dishes to enhance the total capacity while maintaining high service QoS?

    \item How do we design an incentive scheme for commercial dish and LEO satellite operators to participate in \ourFrameworkName?
\end{enumerate}

Existing studies focused on in-space data routing to address the sustainability issues \cite{yang_towards_2016, chou_sustainable_2022, chou_fcsn_2024} and improve QoS \cite{huang_Pheromone_2022}. Different from their work, we focus on utilizing existing commercial dishes and provide incentives to compensate the cost for commercial dish providers. Some other studies focused on the integration of space and terrestrial networks \cite{li_service_2018, chen_time_varying_2021, zhang_enabling_low_latency_2022, li_stable_hier_2024}, while we consider the sustainability issues for LEO satellite networks and incentives for the existing commercial dish providers in terrestrial networks to participate.

To the best of our knowledge, we are the first to systematically study the problem of offloading data from LEO satellites to existing commercial dishes while considering orchestrating collaboration, maintaining high service QoS, and incentives for dish providers and LEO satellite operators. The contributions of this paper are summarized below:

\begin{itemize}
\item We propose a novel sustainable and cost-effective offloading framework named Sustainable Collaborative Offloading (\ourFrameworkName), which orchestrates the collaboration between commercial dishes to receive offloading data from LEO satellite networks.

\item We propose two algorithms, \emph{collaborator group set construction algorithm} and \emph{collaborator selection and total payment algorithm}, to construct candidate groups of commercial dishes for higher capacity and determine payments for incentives, respectively.

\item We perform extensive simulations based on real-world settings. The results show that our \ourFrameworkName{} reduces $\energyReductionPercentageLow\%$ to $\energyReductionPercentageHigh\%$ more energy consumption, $\lifeConsumptionReductionPercentageLow\%$ to $\lifeConsumptionReductionPercentageHigh\%$ more life consumption, and  $\delayReductionPercentageLow\%$ to $\delayReductionPercentageHigh\%$ more end-to-end latency on average, compared to the other three state-of-the-art schemes.
\end{itemize}

The remainder of this paper is organized as follows: we discuss background and motivation in Section \ref{sec:background_and_motivation} and \ourFrameworkName{} framework and models in Section \ref{sec:system_model}. The problem formulation and solution are presented in Section \ref{sec:problem_formulation_and_solution}, followed by the performance evaluation in Section \ref{sec:eval}. The related work is discussed in Section \ref{sec:related_work}, future work is discussed in Section \ref{sec:future_work}, and the conclusion is in Section \ref{sec:conclusion}.

%% file: sections/background_motivation.tex
\section{Background and Motivation}
\label{sec:background_and_motivation}

\begin{figure}[t]
    \mbox{
        \begin{minipage}[t]{0.49\linewidth}
            \centering
            \includegraphics[width=\linewidth]{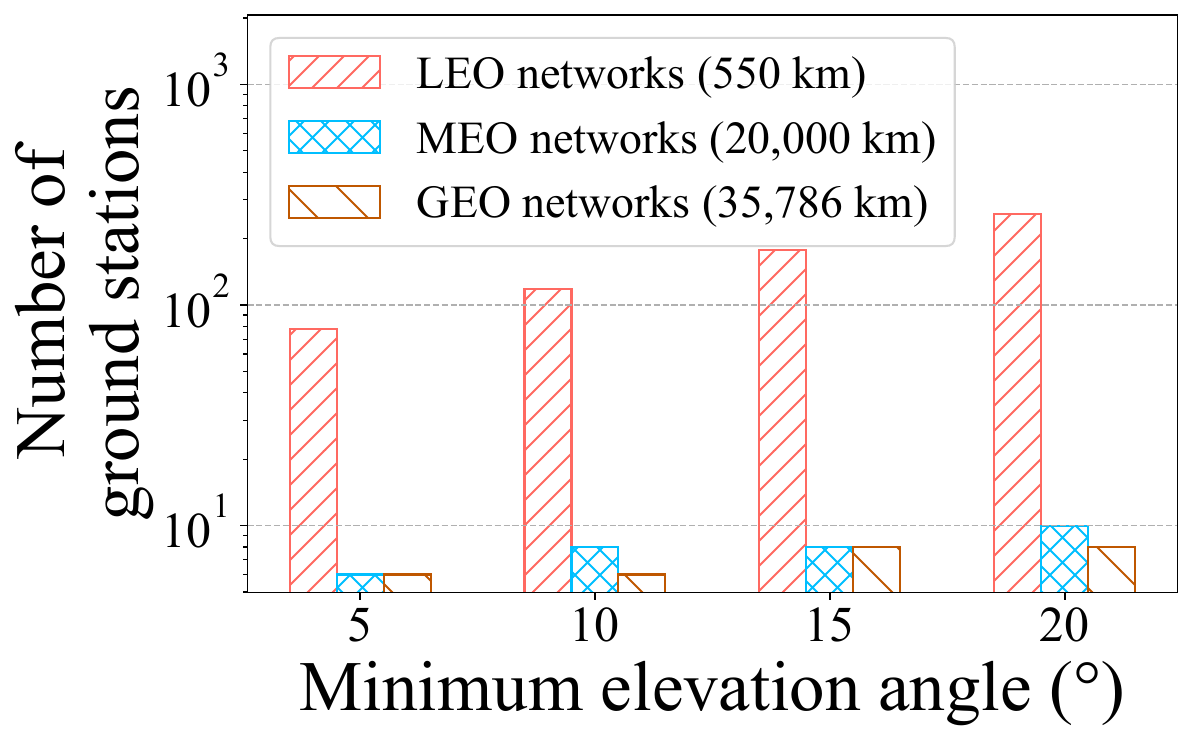} 
            \caption{Required ground stations for bent-pipe.} \label{fig:num_gs_required_bent_pipe}
        \end{minipage}
        \begin{minipage}[t]{0.49\linewidth}
            \centering
            \includegraphics[width=\linewidth]{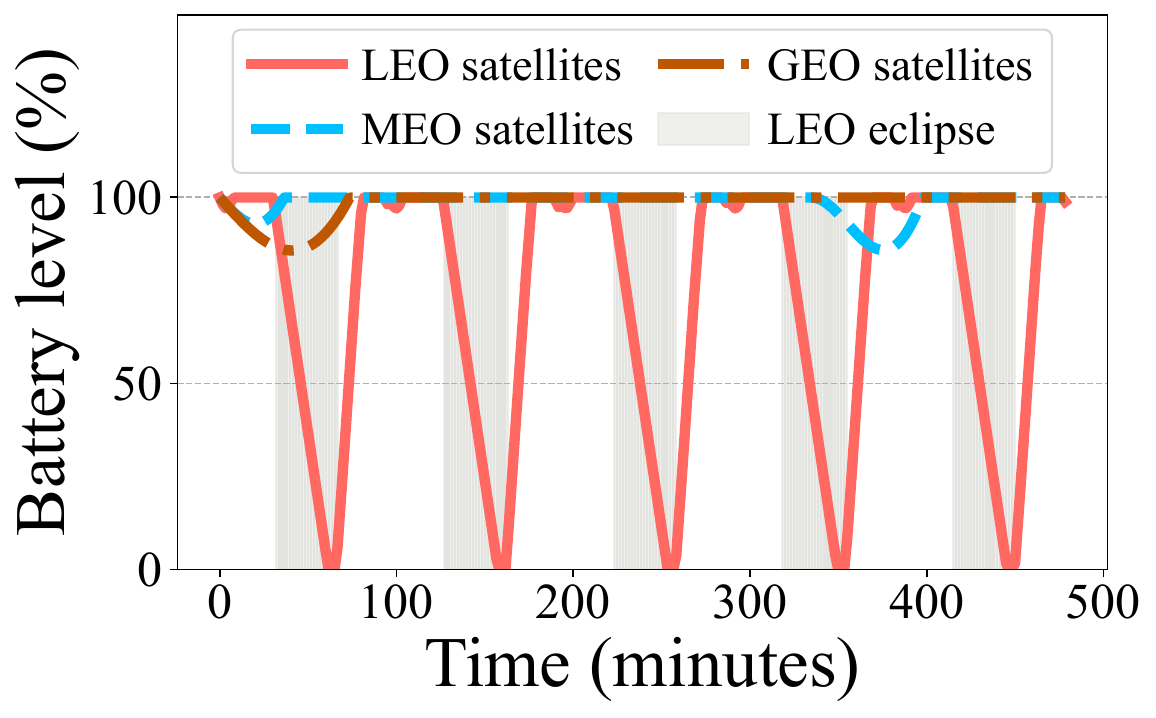} 
            \caption{Satellite battery level over time.} \label{fig:battery_level_changes}
        \end{minipage}
    } 
    \vspace{-3mm}
\end{figure}

\begin{figure}[t]
\centerline{\includegraphics[width=0.43\textwidth]{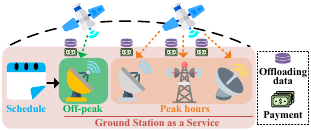}}
\caption{Overview of Ground Station as a Service.}
\label{fig:gsaas}
\vspace{-6mm}
\end{figure}

\subsection{Why Solely Using Bent-Pipe Is Not Practical?}

In the bent-pipe strategy, each LEO satellite relays data from one location to another location which requires the two location points to be under the coverage of the same LEO satellite \cite{ma_starlink_measure_2023}. This requires a dense ground station deployment, which is not a practical routing strategy for global coverage.

LEO satellite networks require more ground stations than the other satellite network family members, Medium Earth Orbit (MEO) and Geosynchronous (GEO) satellite networks, for bent-pipe routing. For example, LEO satellites with altitudes of $550$ km to $1,000$ km can only cover $1.70\%$ to $3.53\%$ of the total Earth surface area at any location \cite{deng_how_many_leo_2021}. We analyze the number of ground stations required for the bent-pipe routing strategy to achieve global coverage. As shown in Fig. \ref{fig:num_gs_required_bent_pipe}, LEO satellite networks require more than a magnitude number of ground stations than MEO and GEO satellite networks in different minimum elevation angles, which are the minimum possible angles between the horizon and the communicating satellite for a ground station.

Considering the high cost (i.e., construction, licensing application, and maintenance costs) of a ground station \cite{vas_l2d2_2021}, the bent-pipe routing strategy is not practical for many LEO satellite network operators, especially those at the very beginning stage with a limited budget. A more cost-effective way to access ground stations is necessary to help these LEO satellite network operators enter the market.

\subsection{Why Solely Using ISL Is Not Sustainable?}

Routing data \emph{solely} through inter-satellite links (ISLs) can decrease satellite battery levels and negatively impact the battery lifespan as it increases traffic load and energy consumption on the satellites. Each satellite hop on the path needs to receive, process, and forward the data to the next hop.

The negative impact on battery lifespan is more severe in LEO satellite networks due to more frequent eclipse events. We analyze the battery level changes over 8 hours for LEO, MEO, and GEO satellites based on the energy-related models from \cite{yang_towards_2016}. As shown in Fig. \ref{fig:battery_level_changes}, LEO satellites encounter eclipse events more than MEO and GEO satellites do due to their low altitudes, resulting in low battery levels several times.

Considering the negative impact on battery lifespan from the energy consumption at low battery levels \cite{yang_towards_2016}, \emph{solely} using ISLs for data routing is not sustainable. Having a cost-effective way to access ground stations is a promising solution to offload data to terrestrial networks, reducing the traffic load and prolonging the battery lifespan for LEO satellites.

\subsection{Ground Station as a Service Still Faces Challenges}
Ground Station as a Service (GSaaS) sells commercial ground station resources, i.e., the use of ground stations for receiving or sending data from or to satellites, in a \emph{pay-as-you-go} manner, where satellite operators only pay for what they have used. Fig. \ref{fig:gsaas} illustrates an overview of GSaaS. LEO satellite operators can schedule ground station resources from GSaaS providers based on their current requirements, i.e., peak or off-peak hours. Then, the providers can allocate corresponding resources adaptively to provide the capabilities that meet the requirements. Next, the data are offloaded to the scheduled ground stations, and the satellite operators pay for the resources they have used.

The pay-as-you-go manner can greatly reduce the burden of initial investment on ground stations for satellite operators, especially those at the beginning stage \cite{vas_l2d2_2021}. Furthermore, using GSaaS with pay-as-you-go can avoid the issue of excessive or insufficient ground stations since satellite operators usually need to plan for ground station construction ahead of time. With these advantages, GSaaS has gained much attention from the industry. For example, NASA plans to switch to GSaaS rather than building dedicated ground stations since the costs of maintaining and operating them consume a significant portion of its budget \cite{NASA2024_commercial_gs}.

However, GSaaS still faces unique challenges. While GSaaS offers flexible pay-as-you-go services, selecting the GSaaS provider that best suits the requirements can be challenging. Using GSaaS from a single provider may lead to budget inefficiency if the resources are overpriced, or not satisfying the requirements if the quality of service (QoS) does not meet the expectation. Using GSaaS from diversified providers is ideal but selecting a group of GSaaS providers is non-trivial since each provider offers services with different QoS, which can result in insufficient or excessive resources. Hence, to fully leverage the potential of GSaaS, we propose \ourFrameworkName{} framework for LEO satellite data offloading by adaptively orchestrating a group of commercial ground stations, while also ensuring service QoS.

%% file: sections/system_model.tex
\section{\ourFrameworkName{} Framework and Models}
\label{sec:system_model}

\subsection{\ourFrameworkName{} Framework Overview}

\textbf{Network scenario.} In \ourFrameworkName, we focus on the scenario consisting of a LEO satellite network and terrestrial networks, where the terrestrial networks can receive data from the LEO satellite network. The LEO satellites are inter-connected with inter-satellite links (ISLs) and the data can be routed via ISLs or ground-satellite links (GSLs). The terrestrial networks include commercial shared resources such as ground stations and 5G base stations equipped with 5G New Radio \cite{3GPP_TR_38811}, providing Ground Station as a Service (GSaaS) (i.e., AWS Ground Station and Azure Orbital Ground Station). We refer to these commercial shared resources as commercial dishes for brevity.\footnote{We use commercial dishes and dishes interchangeably in this paper.}

\textbf{\ourFrameworkName{} overview.} Fig. \ref{fig:vision} illustrates a scenario for \ourFrameworkName, where a user terminal generates and transmits data to the LEO satellite A through a GSL. Then, the data are routed via ISL to the LEO satellite B, which runs an auction on the virtual auction platform to determine a group of commercial dishes for offloading.  Next, the data are offloaded to a group of two collaborative commercial dishes, i.e.,  a ground station and a 5G New Radio base station, under the coverage of the LEO satellite B in terrestrial networks through a GSL. Finally, the data are further routed to the destination through terrestrial networks. We refer to the LEO satellite B as the offloading satellite, where the offloading satellite is the LEO satellite that offloads data.

\textbf{\ourFrameworkName{} network model.} We use $\mathcal{G} = \{\mathcal{V} =  \mathcal{V}_s \cup \mathcal{V}_g, \mathcal{E} = \mathcal{E}_s \cup \mathcal{E}_g\}$ to denote the topology of the LEO satellite and terrestrial networks, where $\mathcal{V}_s=\{v_i|i=1,\dots,|\mathcal{V}_s|\}$ is the set of LEO satellites, $\mathcal{V}_g=\{v_k|k=1,\dots,|\mathcal{V}_g|\}$ is the set of commercial dishes, $\mathcal{E}_s$ is the set of ISLs, and $\mathcal{E}_g$ is the set of GSLs. We adopt the discrete-time model to split the period of the LEO satellite network into a set of intervals, which can be represented as $T=\{\tau\}$. During each interval, the topology of the LEO satellite network is assumed to be stable.

\begin{figure}[t]
\centerline{\includegraphics[width=0.5\textwidth]{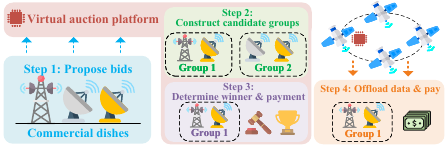}}
\caption{\small{The process of a reverse auction.}}
\label{fig:reverse_auction_collaboration_group}
\vspace{-3mm}
\end{figure}

\subsection{Reverse Auction Platforms and Collaborative Dishes}
In \ourFrameworkName, there exist some \emph{virtual auction platforms} on LEO satellites and each of them runs a reverse auction to select groups of collaborative dishes as winners to receive offloading data. Since dishes are operated by different providers, we use reverse auctions for dish providers to propose bids because only they have their cost information. Fig. \ref{fig:reverse_auction_collaboration_group} illustrates the process of a reverse auction with the following steps:

\begin{enumerate}[label=\circled{\arabic*}, noitemsep,left=0pt]
    \item \textbf{Proposing bids to tasks.} During each interval $\tau$, each auction platform collects nearby tasks and publishes them. Then, the dishes under the satellite coverage propose bids to the auction platform, where each bid contains the cost, the latency for routing data to the destination, and the dish capacity information.
    \item  \textbf{Constructing candidate groups.} After receiving bids, each auction platform constructs a set of candidate groups before selecting a winner. Each candidate group contains at least one dish. The capacity of each dish is considered during the construction to ensure the total capacity of each group satisfies the requirements.
    \item \textbf{Selecting winners and determining payments.} Then, each auction platform selects a winner from the set of candidate groups for each task it receives. Each winner is a group of collaborative dishes to receive the offloading data. We define the group of collaborative dishes as a collaborator group, which contains at least one dish. Next, the payment for the winner is determined by the auction platform.

    \item \textbf{Offloading data to winners.}  Next, each offloading satellite orchestrates the collaboration between the dish(es) in the winner for each task. The data traffic is split and offloaded to each dish in the winner at the beginning of the next interval. The payment will be made to each dish in the winner only if the data offloading succeeds.
\end{enumerate}

We assume each LEO satellite is equipped with a processing device \cite{zdnet_2023_spacex} which can host a virtual auction platform when needed. Each task can only be sent to one auction platform for being processed.

\textbf{Data offloading model.} The LEO satellite network has tasks that require data to be offloaded. Our \ourFrameworkName{} selects a group of commercial dishes as the collaborator group and orchestrates the collaboration between them to receive the offloading data for each task. The orchestrated collaboration can increase the total capacity to avoid the issue where a single commercial dish may not have sufficient capacity for the tasks. Formally, we denote the set of tasks on the auction platform run by satellite $v_i\in \mathcal{V}_s$ during $\tau$ as 
\begin{equation}
\label{eq:task_set}
    \mathbb{T}_i(\tau) = \{t_{i,j}^\tau=(\mathcal{D}_{i,j}^\tau, \mathcal{B}_{i,j}^\tau, \mathcal{T}_{i,j}^\tau)|j=1,\dots,|\mathbb{T}_i(\tau)|\}, 
\end{equation}
where $t_{i,j}^\tau$ is the $j$th task on the satellite $v_i$, and $\mathcal{D}_{i,j}^\tau$, $\mathcal{B}_{i,j}^\tau$,  and $\mathcal{T}_{i,j}^\tau$ are the QoS requirements of $t_{i,j}^\tau$ in terms of delay, bandwidth, and offloading data amount, respectively. 

Similarly, we define the set of bids collected by the auction platform run by $v_i$ for task $t_{i,j}^\tau$ during interval $\tau$ as 
\begin{equation}
\label{eq:bid_definition}
    \hspace{-0.5mm} B_{i,j}(\hspace{-0.5mm}\tau\hspace{-0.5mm})\hspace{-0.5mm}=\hspace{-0.5mm}\{b_k(\hspace{-0.5mm}\tau\hspace{-0.5mm})\hspace{-0.5mm}=\hspace{-0.5mm}(\xi_k, \gamma_k, \Delta_k,  c_k, v_{i^{\prime}})|v_k \hspace{-0.5mm}\in \hspace{-0.5mm}C_i(\hspace{-0.5mm}\tau\hspace{-0.5mm}) \hspace{-0.1mm}\cap\hspace{-0.1mm} C_{i'}(\hspace{-0.5mm}\tau\hspace{-0.5mm})
\},
\end{equation}
where $b_k(\tau)$ is the bid proposed by the commercial dish $v_k \in C_i(\tau) \cap C_{i'}(\tau)$, and $C_i(\tau)$ and $C_{i^{\prime}}(\tau)$ are the sets of dishes that are within the coverage of $v_i$ and $v_{i^{\prime}}$ during $\tau$ and can maintain a GSL with $v_i$ or $v_{i^{\prime}}$ during offloading, respectively. Each $b_k(\tau)$ contains the estimated latency $\xi_k$ for routing data to the destination, dish bandwidth $\gamma_k$, offloading data amount $\Delta_k$, the cost $c_k$, and the offloading satellite $v_{i^{\prime}}$. 

The calculation of the cost $c_k$ will be discussed in Section \ref{sec:dish_cost}. Since the satellite $v_i$ may be moving away from the dish $v_k$, the dish can optionally specify an offloading satellite $v_{i^{\prime}}$ in the bid $b_k(\tau)$ for a longer GSL connection time. Otherwise, $v_{i^{\prime}}$ and $v_i$ can be the same satellite. Without loss of generality, we assume $v_{i^{\prime}} = v_i$ for the discussion of our models.

We denote $\mathbb{G}_{i,j}(\tau)$ as the \emph{collaborator group set} for task $t_{i,j}^\tau$ during interval $\tau$. It can be constructed by the auction platform run by the satellite $v_i$, which is defined as
\begin{equation}
\mathbb{G}_{i,j}(\tau)=\{g_\lambda(\tau)|g_\lambda(\tau)\subseteq B_{i,j}(\tau)\},
\end{equation}
where $g_\lambda(\tau)$ is the $\lambda^{th}$ collaborator group in the collaborator group set. It is a subset of $B_{i,j}(\tau)$. The construction of the collaborator group set will be discussed in Section \ref{sec:collaborator_group_set_construction_algo}.

The data can be routed on the original or offloading path. Their formal definitions are as follows:

\begin{definition}
Given a task $t_{i,j}^\tau \in \mathbb{T}_i(\tau), \forall v_i\in \mathcal{V}_s$, we define $p^{SAT}_{i, j}(\tau)$ as the \emph{original path} from the source to destination LEO satellite without offloading, which is an ordered set with $s_{m,j}^{SAT}$ as the $m^{th}$ satellite on the path. We have 
\begin{equation}
    p^{SAT}_{i, j}(\tau) = (s_{m,j}^{SAT}|s_{m,j}^{SAT}\in \mathcal{V}_s).
\end{equation}
\end{definition}

\begin{definition}
Given a task $t_{i,j}^\tau \in \mathbb{T}_i(\tau), \forall v_i\in \mathcal{V}_s$, we define $p^{GRD}_{i, j\rightarrow k}(\tau)$ as the \emph{offloading path} from the source to the offloading LEO satellite $v_i$, where the data are offloaded to the dish $v_k$ from $v_i$. The offloading path is an ordered set with  $s_{m,j}^{GRD}$ as the $m^{th}$ satellite on the path. We have 
\begin{equation}
p^{GRD}_{i, j\rightarrow k}(\tau)= (s_{m,j}^{GRD}|s_{m,j}^{GRD}\in \mathcal{V}_s).
\end{equation}
\end{definition}

Since data are offloaded early before the destination satellite, we have $p^{GRD}_{i, j\rightarrow k}(\tau)\subset p^{SAT}_{i, j}(\tau)$ and $|p^{GRD}_{i, j\rightarrow k}(\tau)|<| p^{SAT}_{i, j}(\tau)|,$ $\forall v_i\in \mathcal{V}_s, t_{i,j}^\tau \in \mathbb{T}_i(\tau), b_k(\tau)\in B_{i,j}(\tau)$.

\subsection{Dish Offloading Cost Model}
\label{sec:dish_cost}
We consider the offloading data amount and bandwidth reservation time to model the offloading cost of the dish $v_k$, which is defined as
\begin{equation}
     c_k = \alpha \Delta_k + \beta\mathcal{R}_k,
\end{equation}
where $\Delta_k$ is the offloading data amount, $\mathcal{R}_k$ is the bandwidth reservation time of the dish $v_k$, $\alpha$ and $\beta$ are their unit prices. The bandwidth reservation time depends on the time the dish takes to receive the data amount with its bandwidth.

Dishes consume resources for receiving, buffering, and processing data, as well as for allocating bandwidth-related resources for offloading. They are assumed to have a stable energy supply without the need for batteries.

\subsection{Offloading Utility Model}

\subsubsection{Total Offloading Utility}
We define the total offloading utility for task $t_{i,j}^\tau$ if data are offloaded from the offloading satellite $v_i\in \mathcal{V}_s$ to the collaborator group $g_{\lambda}(\tau)\in \mathbb{G}_{i,j}(\tau)$ as
\begin{equation}
     \overline{U}_{i,j\rightarrow \lambda}^\tau = w_1U_{i,j\rightarrow \lambda}^E(\tau) + w_2U_{i,j\rightarrow \lambda}^D(\tau) + w_3U_{i,j\rightarrow \lambda}^L(\tau),
\end{equation}
where $U_{i,j\rightarrow \lambda}^E(\tau)$, $U_{i,j\rightarrow \lambda}^D(\tau)$, and $U_{i,j\rightarrow \lambda}^L(\tau)$ are the normalized utility of the reduced energy consumption, improved end-to-end latency, and extended service life gained from offloading data for the task $t_{i,j}^\tau$ to the collaborator group $g_{\lambda}(\tau)$, respectively, and $w_1$, $w_2$, and $w_3$ are the corresponding weights, which can be adjusted for different requirements.

\subsubsection{Utility of Reduced Energy Consumption}

The reduced energy consumption is achieved by offloading data to the collaborator group $g_{\lambda}(\tau)$, allowing the LEO satellite hops after the offloading satellite to save energy. We define the normalized utility of the reduced energy consumption as
\begin{equation}
    U_{i,j\rightarrow \lambda}^E(\tau) = \frac{E_{i,j}^{SAT}(\tau)-E_{i,j\rightarrow \lambda}^{GRD}(\tau)}{E_{i,j}^{SAT}(\tau)},
\end{equation}
where $E_{i,j}^{SAT}(\tau)$ and $E_{i,j\rightarrow \lambda}^{GRD}(\tau)$ are the energy consumption of the LEO satellites for routing on the original path and offloading path with the collaborator group $g_{\lambda}(\tau)$, respectively. The fixed energy consumption for maintaining satellite operation is omitted in $U_{i,j\rightarrow \lambda}^E(\tau)$ as it is amortized by multiple tasks.

We define the energy consumption of the LEO satellites for routing data on the original path $p^{SAT}_{i, j}(\tau)$ as
\begin{equation}
E_{i,j}^{SAT}(\tau) = \sum\nolimits_{s_{m,j}^{SAT} \in p^{SAT}_{i, j}(\tau)} \mathcal{T}_{i,j}^\tau\epsilon_m,
\end{equation}
where $\epsilon_m$ is the energy consumption of a unit data for $s_{m,j}^{SAT}$. 

Conversely, the energy consumption of the LEO satellites for routing data on the offloading path $p^{GRD}_{i, j\rightarrow k}(\tau)$ with each dish $v_k \in \mathcal{V}_g$ associated with $b_k(\tau)$ in the collaborator group $g_{\lambda}(\tau)$ is
\begin{equation}
E_{i,j\rightarrow \lambda}^{GRD}(\tau) = \sum\nolimits_{b_k(\tau)\in g_{\lambda}(\tau)} \sum\nolimits_{s_{m,j}^{GRD}\in p^{GRD}_{i, j\rightarrow k}(\tau)} \mathcal{T}_{i,j\rightarrow k}^\tau \cdot \epsilon_m,
\end{equation}
where $\mathcal{T}_{i,j\rightarrow k}^\tau$ is the split data traffic to the dish $v_k$ and we have $\sum\nolimits_{b_k(\tau)\in g_{\lambda}(\tau)} \mathcal{T}_{i,j\rightarrow k}^\tau = \mathcal{T}_{i,j}^\tau$.

\subsubsection{Utility of Improved End-to-End Latency}

The improved end-to-end latency is gained from routing data on the offloading path since LEO satellite networks are relatively more congested than the terrestrial networks \cite{li_leotp_2023}, which can be calculated from the latency difference between routing on the original and offloading paths. We define the normalized utility of the improved end-to-end latency as

\begin{equation}
    U_{i,j\rightarrow \lambda}^D(\tau) = \frac{D_{i,j}^{SAT}(\tau)-D_{i,j\rightarrow \lambda}^{GRD}(\tau)}{D_{i,j}^{SAT}(\tau)},
\end{equation}
where $D_{i,j}^{SAT}(\tau)$ and $D_{i,j\rightarrow \lambda}^{GRD}(\tau)$ are the end-to-end latency for routing data on the original and offloading path, respectively.

The end-to-end latency $D_{i,j}^{SAT}(\tau)$ for routing data on the original path is defined as 
\begin{equation}
D_{i,j}^{SAT}(\tau) = \sum\nolimits_{m=1}^{|p^{SAT}_{i, j}(\tau)| - 1} l_{m, m+1},
\end{equation}
where $l_{m, m+1}$ is the sum of propagation, queuing, and transmission delay for routing data from $s_{m,j}^{SAT}$ to $s_{m+1,j}^{SAT}$ on the original path $p^{SAT}_{i, j}(\tau)$. 

Conversely, the end-to-end latency $D_{i,j\rightarrow \lambda}^{GRD}(\tau)$ for routing data on the offloading path to the collaborator group $g_\lambda(\tau)$ is the maximum latency among the offloading paths with the dishes in the collaborator group
\begin{equation}
\label{eq:delay_offloading}
D_{i,j\rightarrow \lambda}^{GRD}(\tau) = \max\nolimits_{b_k(\tau)\in g_\lambda(\tau)}D_{i,j\rightarrow k}^{GRD}(\tau),
\end{equation}
where $D_{i,j\rightarrow k}^{GRD}(\tau)$ is the end-to-end latency for routing data on the offloading path to the dish $v_k$. We have
\begin{equation}
\label{eq:delay_offload_single_dish}
D_{i,j\rightarrow k}^{GRD}(\tau) = \sum\nolimits_{m=1}^{|p^{GRD}_{i, j\rightarrow k}(\tau)| - 1} \hspace{-1mm}l_{m, m+1} + l_{i, k} + \xi_{k},
\end{equation}
where $l_{i, k}$ is the latency for offloading data from the offloading satellite $v_i$ to the dish $v_k$, and $\xi_{k}$ is the estimated latency for routing data from the dish $v_k$ to the destination.

\subsubsection{Utility of Extended Service Life}

We focus on the battery lifespan for satellite service life since batteries are essential for satellite operation. We extend the battery life consumption model from \cite{yang_towards_2016} by considering the remaining battery lifespan. We define the service life cost for the LEO satellite $s_{m,j}^{SAT}$ (resp. $s_{m,j}^{GRD}$) on the original path $p^{SAT}_{i, j}(\tau)$ (resp. the offloading path $p^{GRD}_{i, j\rightarrow k}(\tau)$) as
\begin{equation}
L_{{m,j}} = \sum\nolimits_{b_k(\tau)\in g_{\lambda}(\tau)} K_{i,j \rightarrow k}^{m} \cdot e^{\frac{Q_m^{max} - Q_m}{Q_m}},
\end{equation}
where $K_{i,j \rightarrow k}^{m}$ is the battery life consumption \cite{yang_towards_2016} from the split data traffic $\mathcal{T}_{i,j\rightarrow k}^\tau$ on the satellite $s_{m,j}^{SAT}$ (resp. $s_{m,j}^{GRD}$), and $Q_m$ and $Q_m^{max}$ are the remaining and maximum battery lifespan on the satellite $s_{m,j}^{SAT}$ (resp. $s_{m,j}^{GRD}$), respectively. The service life cost increases when the remaining battery lifespan is lower and we define the cost as infinity if $Q_m=0$.

The life consumption model $K_{i,j \rightarrow k}^{m}$ is defined as $\int_{\psi^{m, 1}_{i, j \rightarrow k}}^{\psi^{m, 2}_{i, j \rightarrow k}} 10^{-A_{m}\psi}(1+A_{m}\ln{10\cdot} (1-\psi)) \, d\psi$ according to \cite{yang_towards_2016}, where $\psi^{m, 1}_{i, j \rightarrow k}$ and $\psi^{m, 2}_{i, j \rightarrow k}$ are the battery level percentages of the satellite $s_{m,j}^{SAT}$ (resp. $s_{m,j}^{GRD}$) before and after the energy consumption for routing the split data traffic $\mathcal{T}_{i,j\rightarrow k}^\tau$ of the task $t_{i,j}^\tau$, and $A_m$ is a constant depending on the battery model on the satellite $s_{m,j}^{SAT}$ (resp. $s_{m,j}^{GRD}$). Note that $K_{i,j \rightarrow k}^{m} = 0$ if the battery level percentage does not decrease \cite{yang_towards_2016}, i.e., $\psi^{m, 1}_{i, j \rightarrow k} \le \psi^{m, 2}_{i, j \rightarrow k}$, because we assume there is no energy consumption from the batteries. 

The utility of the extended service life for the task $t_{i,j}^\tau$ is gained from the reduced service life cost from the LEO satellite hops after the offloading satellite for each split data traffic. Thus, we have
\begin{equation}
 \hat{U}_{i,j\rightarrow \lambda}^L(\tau) = \sum\nolimits_{b_k(\tau)\in g_{\lambda}(\tau)} \sum\nolimits_{m=|p^{GRD}_{i, j\rightarrow k}(\tau)| + 1}^{|p^{SAT}_{i, j}(\tau)|} L_{{m,j}},
\end{equation}
and then we normalize it by the service life cost for routing data on the original path
\begin{equation}
\label{eq:normalized_utility_extended_service_life}
    U_{i,j\rightarrow \lambda}^L(\tau) = \frac{\hat{U}_{i,j\rightarrow \lambda}^L(\tau)}{\sum\nolimits_{m=1}^{|p^{SAT}_{i, j}(\tau)|} L_{{m,j}}}.
\end{equation}

\subsection{Dish Reliability-Based Utility Adjustment}

Each dish $v_k$ may fail to receive the offloading data due to several reasons, such as the arrival of unexpected traffic burst and hardware failure, which can cause a significant loss for the LEO satellite operator since the offloading satellite $v_i$ needs to re-transmit the data. To address this issue, we discount the total offloading utility of the collaborator group $g_{\lambda}(\tau)$ by considering the failure rate of each dish $v_k$ associated with $b_k(\tau)\in g_{\lambda}(\tau)$ to lower the chance of winning in future intervals. The discounted total offloading utility is defined as
\begin{equation}
\label{eq:discounted_utility}
U_{i,j\rightarrow \lambda}^\tau = \overline{U}_{i,j\rightarrow \lambda}^\tau \prod\nolimits_{b_k(\tau) \in g_{\lambda}(\tau)}(1-f_k^{\tau}),
\end{equation}
where $f_k^{\tau}$ is the failure probability in the \emph{past} history of the dish $v_k$. We initialize $f_k^{0} = 0$ for the first interval and update it for each subsequent interval $\tau$ as
\begin{equation}
f_k^{\tau} = 
\begin{cases} 
\label{eq:failure_update}
  \frac{f_k^{\tau - 1} n_k + \sigma_f}{n_k + 1}, & \text{if $v_k$ wins a bid in interval $\tau-1$}, \\
  f_k^{\tau - 1}, & \text{otherwise},
\end{cases}
\end{equation}
where $n_k$ is the number of times that the dish $v_k$ wins a bid in a collaborator group and $\sigma_f \in \{0, 1\}$ is a binary variable indicating whether the dish $v_k$ wins a bid but fails to receive the offloading data in the interval $\tau-1$ (i.e., $\sigma_f=1$) or not.

\section{Problem Formulation and Solution}
\label{sec:problem_formulation_and_solution}

\subsection{Collaborative Dishes as Ground Stations Problem}
The objective of the Collaborative Dishes as Ground Stations (CDGS) problem is to maximize the total offloading utility with the joint consideration of QoS requirements for the tasks and their budget. We formulate the CDGS problem as an optimization problem with Eq. (\ref{goal}) as the objective
\begin{equation}
\label{goal}
\textbf{maximize} \sum\nolimits_{t_{i,j}^\tau \in \mathbb{T}_i(\tau), g_{\lambda}(\tau) \in \mathbb{G}_{i,j}(\tau)} U_{i,j\rightarrow \lambda}^\tau x_{i,j\rightarrow \lambda}^\tau,
\end{equation}
\begin{equation*}
s.t. \quad (\ref{eq:task_set})-(\ref{eq:normalized_utility_extended_service_life}),
\end{equation*}
\begin{equation}
\label{cons0}
x_{i,j\rightarrow \lambda}^\tau \in \{0,1\},
\end{equation}
\begin{equation}
\label{cons1}
\sum\nolimits_{ g_{\lambda}(\tau)\in \mathbb{G}_{i,j}(\tau)} x_{i,j\rightarrow \lambda}^\tau \le 1, \forall t_{i,j}^\tau \in \mathbb{T}_i(\tau),
\end{equation}
\begin{equation}
\label{cons2}
D_{i,j\rightarrow \lambda}^{GRD}(\tau) \le \mathcal{D}_{i,j}^\tau, \forall x_{i,j\rightarrow \lambda}^\tau = 1, t_{i,j}^\tau \in \mathbb{T}_i(\tau),
\end{equation}
\begin{equation}
\label{cons:bandwidth}
 B_{i,j\rightarrow \lambda}^{GRD}(\tau) \ge \mathcal{B}_{i,j}^\tau, \forall x_{i,j\rightarrow \lambda}^\tau = 1, t_{i,j}^\tau \in \mathbb{T}_i(\tau),
\end{equation}
\begin{equation}
\label{cons2_2}
 \sum\nolimits_{b_k(\tau)\in g_\lambda(\tau)}\Delta_k \ge \mathcal{T}_{i,j}^\tau, \forall x_{i,j\rightarrow \lambda}^\tau = 1, t_{i,j}^\tau \in \mathbb{T}_i(\tau),
\end{equation}
\begin{equation}
\label{cons3}
\hspace{-2mm}x_{i,j\rightarrow \lambda}^\tau \sum\nolimits_{b_k(\tau)\in g_{\lambda}(\tau)}\hspace{-0.5mm} p_{ij\rightarrow k}^\tau \hspace{-0.5mm}\le \hspace{-0.5mm}\beta^\tau, \forall t_{i,j}^\tau \hspace{-0.5mm}\in \hspace{-0.5mm}\mathbb{T}_i(\tau), v_i\hspace{-0.5mm}\in \hspace{-0.5mm}\mathcal{V}_s,
\end{equation}
\begin{equation}
\label{cons4}
\hspace{-2mm} x_{i,j\rightarrow \lambda}^\tau c_k \le p_{ij\rightarrow k}^\tau , \forall t_{i,j}^\tau \hspace{-0.5mm}\in\hspace{-0.5mm} \mathbb{T}_i(\tau), v_i\hspace{-0.5mm}\in \hspace{-0.5mm}\mathcal{V}_s, b_k(\tau)\hspace{-0.5mm} \in\hspace{-0.5mm} g_{\lambda}(\tau),
\end{equation}
where $x_{i,j\rightarrow \lambda}^\tau$ in the constraint (\ref{cons0}) is a binary variable indicating whether the $\lambda^{th}$ collaborator group is assigned to the task $t_{i,j}^\tau$ for offloading ($x_{i,j\rightarrow \lambda}^\tau=1$) or not. The constraint (\ref{cons1}) indicates that at most one collaborator group can be assigned to a task. The constraints (\ref{cons2}), (\ref{cons:bandwidth}), and (\ref{cons2_2}) are the QoS requirements for end-to-end latency, bandwidth, and data amount. The bandwidth $ B_{i,j\rightarrow \lambda}^{GRD}(\tau)$ in constraint (\ref{cons:bandwidth}) is the sum of the dish bandwidth in the collaborator group defined as $\sum_{b_k(\tau)\in g_\lambda(\tau)} \gamma_k$. The constraint (\ref{cons2_2}) guarantees that the collaborator group can handle the required data amount. The constraint (\ref{cons3}) guarantees that the total payment to the assigned collaborator group is within the budget $\beta^\tau$ for the task, and the constraint (\ref{cons4}) guarantees that the payment to each dish in the collaborator group is at least its cost.

\subsection{Collaborator Group Set Construction Algorithm}
\label{sec:collaborator_group_set_construction_algo}
\begin{algorithm}[t]
\caption{Collaborator group set construction}
\label{algo:cgsc}
\begin{algorithmic}[1]
\State \textbf{Input:} $\mathbb{T}_i(\tau), \{B_{i,j}(\tau) | \forall t_{i,j}^\tau\in\mathbb{T}_i(\tau)\}, \mathcal{N}, \mathcal{M}$
\State \textbf{Output:} $\{\mathbb{G}_{i,j}(\tau)| \forall t_{i,j}^\tau\in\mathbb{T}_i(\tau)\}$
\For{\( t_{i,j}^\tau \in \mathbb{T}_i(\tau) \)} 
\label{line:candi_bid_begin}
\State $\mathbb{G}_{i,j}(\tau)=\emptyset, \overline{\mathbb{G}}_{i,j}=\emptyset, \overline{B}_{i,j}(\tau)=\emptyset$ \label{line:init_group_set}
\For{$b_k(\tau)\in B_{i,j}(\tau)$} 
\If{$D_{i,j\rightarrow k}^{GRD}(\tau) \le \mathcal{D}_{i,j}^\tau$}
\State $\overline{B}_{i,j}(\tau)=\overline{B}_{i,j}(\tau) \cup \{b_k(\tau)\}$
\EndIf
\EndFor \label{line:candi_bid_end}
\State $[\overline{\mathbb{G}}_{i,j}]^1=\{\{b_k(\tau)\}| \forall b_k(\tau)\in \overline{B}_{i,j}(\tau)\}$ \label{line:layer_1_candidate_group_sets_1}
\State $\overline{\mathbb{G}}_{i,j}=\overline{\mathbb{G}}_{i,j}\cup [\overline{\mathbb{G}}_{i,j}]^1, n=2$ \label{line:layer_1_candidate_group_sets_2}
\While{$n \le \mathcal{N}$} \label{line:start_layer_group_constuction}
\State Construct $[\overline{\mathbb{G}}_{i,j}]^{n-1}_{\mathcal{M}}$ that contains the top $\mathcal{M}$ groups of $[\overline{\mathbb{G}}_{i,j}]^{n-1}$ with the smallest costs \label{line:construct_M}
\State $[\overline{\mathbb{G}}_{i,j}]^n\hspace{-0.5mm}=\hspace{-0.5mm}\{g_\lambda(\tau) \cup g_{\lambda'}(\tau)| g_\lambda(\tau),g_{\lambda'}(\tau)\hspace{-0.5mm}\in\hspace{-0.5mm}[\overline{\mathbb{G}}_{i,j}]^{n-1}_{\mathcal{M}}\}$ \label{line:merge_layer_groups}
\State $\overline{\mathbb{G}}_{i,j}=\overline{\mathbb{G}}_{i,j}\cup [\overline{\mathbb{G}}_{i,j}]^n$ \label{line:place_layer_n_to_candidate_group}
\State $n=n+1$
\EndWhile \label{line:end_layer_group_constuction}
\For{$g_\lambda(\tau) \in \overline{\mathbb{G}}_{i,j}$} \label{line:traffic_bandwidth_req_check_start}
\If{$\sum_{b_k(\tau)\in g_\lambda(\tau)} \Delta_k < \mathcal{T}_{i,j}^\tau \text{or}$ \\ $\sum_{b_k(\tau)\in g_\lambda(\tau)} \gamma_k < \mathcal{B}_{i,j}^\tau$} \label{line:traffic_amount_bandwidth_check}
\State $\overline{\mathbb{G}}_{i,j}=\overline{\mathbb{G}}_{i,j}-\{g_\lambda(\tau) \}$ \label{line:remove_group_no_satisfy_req}
\EndIf
\EndFor  \label{line:bid_manipulation_end} \label{line:traffic_bandwidth_req_check_end}

\For{$g_\lambda(\tau)\in \overline{\mathbb{G}}_{i,j}$} \label{line:budget_filter_begin}
\If{$\sum_{b_k(\tau)\in g_\lambda(\tau)} c_k \le \beta^\tau$}
\State $\mathbb{G}_{i,j}(\tau)=\mathbb{G}_{i,j}(\tau)\cup \{g_\lambda(\tau)\}$
\EndIf
\EndFor \label{line:budget_filter_end}
\EndFor
\end{algorithmic}
\end{algorithm}

We propose the Collaborator Group Set Construction (CGSC) algorithm as shown in Algorithm \ref{algo:cgsc} to construct a set of collaborator groups as potential candidates before selecting a winner to reduce the problem scale. Given a set of tasks $\mathbb{T}_i(\tau)$ and bids $B_{i,j}(\tau)$, the CGSC algorithm controls the size of the collaborator group set for each task using the parameters $\mathcal{N}$ and $\mathcal{M}$, respectively.

We begin by iterating each task in line \ref{line:candi_bid_begin}. In lines \ref{line:init_group_set}-\ref{line:candi_bid_end}, we initialize the candidate bid set $\overline{B}_{i,j}(\tau)$ for the task $t_{i,j}^\tau$ by considering the individual dish delay for the constraint (\ref{cons2}), in which each dish offloading delay cannot exceed the task delay requirement. Next, in lines \ref{line:layer_1_candidate_group_sets_1}-\ref{line:layer_1_candidate_group_sets_2}, we initialize a candidate group set where each group contains only one bid, denoted as layer-1 candidate group set $[\overline{\mathbb{G}}_{i,j}]^1$, and then we merge it to the total candidate group set $\overline{\mathbb{G}}_{i,j}$.

Next, in lines \ref{line:start_layer_group_constuction}-\ref{line:end_layer_group_constuction}, we recursively construct the layer-n candidate group set $[\overline{\mathbb{G}}_{i,j}]^n$ based on the top $\mathcal{M}$ groups ($\mathcal{M}$ groups with the smallest costs) from the previous layer-(n-1) candidate group set $[\overline{\mathbb{G}}_{i,j}]^{n-1}$, denoted as $[\overline{\mathbb{G}}_{i,j}]^{n-1}_{\mathcal{M}}$. We repeat the process until the layer-$\mathcal{N}$ candidate group set $[\overline{\mathbb{G}}_{i,j}]^{\mathcal{N}}$ is constructed. Specifically, in line \ref{line:construct_M}, the top $\mathcal{M}$ groups with the smallest costs are selected (can be extended for other metrics). From $[\overline{\mathbb{G}}_{i,j}]^{n-1}_{\mathcal{M}}$, each pair of groups are merged to form a new group in $[\overline{\mathbb{G}}_{i,j}]^n$ as shown in line \ref{line:merge_layer_groups}. Then, $[\overline{\mathbb{G}}_{i,j}]^n$ is merged to the total candidate group set $\overline{\mathbb{G}}_{i,j}$ in line \ref{line:place_layer_n_to_candidate_group}.

From lines \ref{line:traffic_bandwidth_req_check_start}-\ref{line:traffic_bandwidth_req_check_end}, we check whether each candidate group in $\overline{\mathbb{G}}_{i,j}$ satisfies the data amount and bandwidth requirement. Then, we remove the group from the candidate group set if it does not satisfy the requirements. Finally, from lines \ref{line:budget_filter_begin}-\ref{line:budget_filter_end}, we perform a preliminary budget check for the cost of each candidate group and merge it to the collaborator group set if the budget can cover the cost. We will perform a budget check for the actual payment in Algorithm \ref{algo:ucb1}.

\begin{figure}[t]
\centerline{\includegraphics[width=0.49\textwidth]{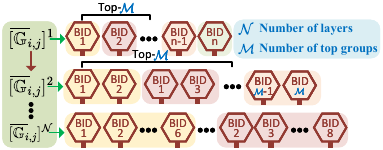}}
\caption{\small{The overview of the $[\overline{\mathbb{G}}_{i,j}]^n$ construction.}}
\label{fig:group_construction}
\end{figure}

Fig. \ref{fig:group_construction} illustrates the overview of the $[\overline{\mathbb{G}}_{i,j}]^n$ construction. We first initialize the layer-1 candidate group set ($[\overline{\mathbb{G}}_{i,j}]^1$) where each bid is a collaborator group. Then, we take each pair of the top-$\mathcal{M}$ groups to form a new group and merge them to get the layer-2 candidate group set ($[\overline{\mathbb{G}}_{i,j}]^2$). We repeat the same process for $\mathcal{N}$ times until we get the layer-$\mathcal{N}$ candidate group set ($[\overline{\mathbb{G}}_{i,j}]^{\mathcal{N}}$).

\subsection{Collaborator Selection and Total Payment Algorithm}
\begin{algorithm}[t]
\caption{Collaborator selection and total payment}
\label{algo:ucb1}
\begin{algorithmic}[1]
\State \textbf{Input:} $\mathbb{T}_i(\tau), \{\mathbb{G}_{i,j}(\tau),\{n_{\lambda} | \forall g_{\lambda}(\tau) \in \mathbb{G}_{i,j}(\tau)\} | \forall t_{i,j}^\tau \in \mathbb{T}_i(\tau)\}$
\State \textbf{Output:} $\{g_\lambda(\tau), \{p_{ij\rightarrow k}^\tau| \forall b_k(\tau)\in g_{\lambda}(\tau)\}| \forall t_{i,j}^\tau \in \mathbb{T}_i(\tau) \}$
\For{\( t_{i,j}^\tau \in \mathbb{T}_i(\tau) \)} \label{line:for_task}
\State $g_{\lambda}(\tau) =\emptyset$
\While{$g_{\lambda}(\tau) =\emptyset~\text{and}~\mathbb{G}_{i,j}(\tau)\neq \emptyset$} \label{line:group_find}
\State $n_{sum} = \ln \sum_{g_\lambda(\tau) \in \mathbb{G}_{i,j}(\tau)}n_{\lambda}$ \label{line:co_begin}
\State $g_{\overline{\lambda}}(\tau) \hspace{-1mm}=\hspace{-1mm}\arg\max_{g_\lambda(\tau)\in \mathbb{G}_{i,j}(\tau), \forall U_{i,j\rightarrow \lambda}^\tau > 0}\\ \hspace{-1mm}\frac{U_{i,j\rightarrow \lambda}^\tau}  {\sum_{b_k(\tau)\in g_\lambda(\tau)} c_k}\hspace{-1mm}+\hspace{-1mm}\sqrt{\frac{2n_{sum}}{n_{\lambda}} }$ \label{line:co_end}
\State  $\mathbb{G}_{i,j}(\tau)=\mathbb{G}_{i,j}(\tau)-\{g_{\overline{\lambda}}(\tau)\}$ \label{line:pay_begin}
\State $n^{\prime}_{sum} = \ln \sum_{g_\lambda(\tau) \in \mathbb{G}_{i,j}(\tau)}n_{\lambda}$
\State $\mathcal{U} = \max_{g_\lambda(\tau)\in \mathbb{G}_{i,j}(\tau), \forall U_{i,j\rightarrow \lambda}^\tau > 0} \\ \frac{U_{i,j\rightarrow \lambda}^\tau }{\sum_{b_k(\tau)\in g_\lambda(\tau)} c_k} +\sqrt{\frac{2n^{\prime}_{sum}}{n_\lambda}}$ \label{line:u_define}
\State $p_{ij\rightarrow \overline{\lambda}}^\tau = \frac{U_{i,j\rightarrow \overline{\lambda}}^\tau + \sum_{b_k(\tau)\in g_{\overline{\lambda}(\tau)}} c_k \sqrt{\frac{2n_{sum}}{n_{\overline{\lambda}}}}}{\mathcal{U}}$ \label{line:pay_end}
\If{$p_{ij\rightarrow \overline{\lambda}}^\tau \le \beta^\tau$} \label{line:pay_alloc_begin}
\State $\beta^\tau = \beta^\tau - p_{ij\rightarrow \overline{\lambda}}^\tau$
\State $g_\lambda(\tau)= g_{\overline{\lambda}}(\tau), n_{\lambda}=n_{\lambda} + 1$
\State $p_{ij\rightarrow k}^\tau = \frac{c_k}{\sum_{b_k(\tau)\in g_{\overline{\lambda}}(\tau)}c_k}p_{ij\rightarrow \overline{\lambda}}^\tau, \forall b_k(\tau)\in g_{\overline{\lambda}}(\tau)$ \label{line:pay_alloc_end}

\EndIf
\EndWhile
\EndFor
\end{algorithmic}
\end{algorithm}

We propose the Collaborator Selection and Total Payment (CSTP) algorithm, as shown in Algorithm \ref{algo:ucb1}, to select a collaborator group as the winner for each task. We take the collaborator group set for each task, i.e., \{$\mathbb{G}_{i,j}(\tau) | \forall t_{i,j}^\tau \in \mathbb{T}_i(\tau)\}$, given by Algorithm \ref{algo:cgsc} and the number of times being selected for each group, i.e., $\{n_{\lambda}| \forall g_{\lambda}(\tau)\in\mathbb{G}_{i,j}(\tau)\}$, as inputs. Then, the CSTP algorithm determines the payment for each dish in the winner collaborator group.

We begin by iterating all the tasks $t_{i,j}^\tau \in \mathbb{T}_i(\tau)$ collected by the auction platform on the offloading satellite $v_i$ in line \ref{line:for_task}. For each task, we iterate the collaborator group set $\mathbb{G}_{i,j}(\tau)$ in line \ref{line:group_find} to select a collaborator group $g_\lambda(\tau)$ for offloading. In lines \ref{line:co_begin}-\ref{line:co_end}, we calculate the total number of times the collaborator groups are selected in log scale and then select the collaborator group, denoted as $g_{\overline{\lambda}}(\tau)$, with the maximum utility-to-cost ratio plus an adjustment factor. The adjustment factor allows the collaborator groups with a higher utility or fewer number of times being selected to be more likely to be selected for exploitation or exploration, respectively.

In lines \ref{line:pay_begin}-\ref{line:pay_end}, we find the second highest utility-to-cost ratio, denoted as $\mathcal{U}$, which is used to calculate the payment for the collaborator group $g_{\overline{\lambda}}(\tau)$. Finally, in lines \ref{line:pay_alloc_begin}-\ref{line:pay_alloc_end}, if the budget is sufficient, we assign $g_{\overline{\lambda}}(\tau)$ to $g_\lambda(\tau)$ as the selected collaborator group, update the number of times being selected $n_\lambda$ and the budget $\beta^\tau$, and allocate the payment for each dish based on the proportion of its cost in the collaborator group. 

%% file: sections/analysis.tex
\section{Theoretical Analysis}
\label{sec:theoretical_analysis}

\begin{definition}
(\textbf{Monotonicity}). For each collaborator group $g_\lambda(\tau), \forall v_i\in \mathcal{V}_s, t_{i,j}^\tau \in \mathbb{T}_i(\tau), g_\lambda(\tau)\in \mathbb{G}_{i,j}(\tau)$ with a set of costs $\{c_k|b_i(\tau)\in g_\lambda(\tau)\}$, if it wins, then another collaborator group $g_{\lambda^\prime}$ with the same total utility and exploration factor but a smaller total cost $\sum_{b_k(\tau)\in g_{\lambda^\prime}(\tau)} c_k$ also wins.
\end{definition} 

\begin{definition}
(\textbf{Critical value}). Given a collaborator group $g_\lambda(\tau), \forall v_i\in \mathcal{V}_s, t_{i,j}^\tau \in \mathbb{T}_i(\tau), g_\lambda(\tau)\in \mathbb{G}_{i,j}(\tau)$ with a set of payments $p_{ij\rightarrow k}^\tau, \forall b_k(\tau)\in g_\lambda(\tau)$, if every dish of this group declares a cost that is no greater than the payment $c_k \le p_{ij\rightarrow k}^\tau$, it must win; otherwise, it loses.
\end{definition}

\begin{definition}
    (\textbf{Individual rationality}). Given a collaborator group $g_\lambda(\tau), \forall v_i\in \mathcal{V}_s, t_{i,j}^\tau \in \mathbb{T}_i(\tau), g_\lambda(\tau)\in \mathbb{G}_{i,j}(\tau)$, our payment strategy satisfies $p_{ij\rightarrow k}^\tau \ge c_k, \forall b_k(\tau) \in g_\lambda(\tau)$.
\end{definition}

\begin{lemma}
\label{lemma-1}
The CSTP algorithm meets monotonicity.
\end{lemma}
\begin{proof}
Assume that the utility and exploration factor of the winner $g_\lambda(\tau)$ are $U_{i,j\rightarrow \lambda}^\tau$ and $\sqrt{\frac{2n_{sum}}{n_{\lambda}}}$, respectively. The total bid cost of $g_\lambda(\tau)$ is $\sum_{b_k(\tau)\in g_\lambda(\tau)} c_k$. If $\exists g_{\lambda^\prime}(\tau)\in \mathbb{G}_{i,j}(\tau) $ that has the same utility and exploration factor, but a smaller sum of cost, i.e., $\sum_{b_{k'}(\tau)\in g_{\lambda^\prime}(\tau)} c_{k'} < \sum_{b_k(\tau)\in g_{\lambda}(\tau)} c_k$, we have $\frac{U_{i,j\rightarrow \lambda}^\tau}  {\sum_{b_k(\tau)\in g_\lambda(\tau)} c_k}+ \sqrt{\frac{2n_{sum}}{n_{\lambda}} }  \hspace{-1.5mm} < \hspace{-1.5mm}\frac{U_{i,j\rightarrow \lambda'}^\tau}  {\sum_{b_{k'}(\tau)\in g_{\lambda^\prime}(\tau)} c_{k'}}+ \sqrt{\frac{2n'_{sum}}{n_{\lambda'}} }$, so $g_{\lambda^\prime}(\tau)$ also wins.
\end{proof}

\begin{lemma}
\label{lemma-2}
The payment of the winner given by the CSTP algorithm is the critical value.
\end{lemma}
\begin{proof}

Given $c_k \hspace{-0.5mm}\le\hspace{-0.5mm} p_{ij\rightarrow k}^\tau, \forall b_k(\tau)\hspace{-0.5mm}\in\hspace{-0.5mm} g_\lambda(\tau)$, based on the definition of $p_{ij\rightarrow k}^\tau$, we have $\frac{c_k}{\sum_{b_k(\tau)\in g_\lambda(\tau)}c_k}p_{ij\rightarrow \lambda}^\tau \ge c_k$. Thus, we get $p_{ij\rightarrow \lambda}^\tau \hspace{-0.5mm}\ge\hspace{-0.5mm} \sum_{b_k(\tau)\in g_\lambda(\tau)}c_k$. Based on Algorithm \ref{algo:ucb1}, we know that $p_{ij\rightarrow \lambda}^\tau \hspace{-0.5mm}= \hspace{-0.5mm}\frac{U_{i,j\rightarrow \lambda}^\tau + \sum_{b_k(\tau)\in g_\lambda(\tau)} c_k \sqrt{\frac{2n_{sum}}{n_{\lambda}}}}{\mathcal{U}} \hspace{-0.5mm}\ge \hspace{-0.5mm}\sum_{b_k(\tau)\in g_\lambda(\tau)}c_k$. After doing some algebraic operations, we get $\frac{U_{i,j\rightarrow \lambda}^\tau}{\sum_{b_k(\tau)\in g_\lambda(\tau)}c_k} + \sqrt{\frac{2n_{sum}}{n_{\lambda}}} \ge \mathcal{U}$. Thus, we know that $g_\lambda(\tau)$ must win. Similarly, if the declared cost is greater than the payment, $g_\lambda(\tau)$ loses in the auction.
\end{proof}

\begin{theorem}
    The CSTP algorithm is truthful. 
    \label{thm2}
\end{theorem}
\begin{proof}
   According to the Myerson’s principle \cite{myerson1981optimal}, if an algorithm is monotone and its payment is the critical value, the algorithm is truthful. Thus, the CSTP algorithm is truthful based on Lemma \ref{lemma-1} and Lemma \ref{lemma-2}.
\end{proof}

\begin{theorem}
    The CSTP algorithm achieves individual rationality for each user dish.
    \label{thm3_0}
\end{theorem}
\begin{proof}
Given a winner $g_{\lambda}(\tau)$, we have $\frac{U_{i,j\rightarrow \lambda}^\tau}  {\sum_{b_k(\tau)\in g_\lambda(\tau)} c_k} + \sqrt{\frac{2n_{sum}}{n_{\lambda}} }   \ge \mathcal{U}$, where $\mathcal{U}$ is defined in line \ref{line:u_define} in Algorithm \ref{algo:ucb1}. Thus, we have  $\frac{U_{i,j\rightarrow \lambda}^\tau+ \sum_{b_k(\tau)\in g_\lambda(\tau)} c_k \sqrt{\frac{2n_{sum}}{n_{\lambda}}}}{\mathcal{U}} \ge \sum_{b_k(\tau)\in g_\lambda(\tau)} c_k$, where the left-hand side is $p_{ij\rightarrow \lambda}^\tau$. Hence, $p_{ij\rightarrow k}^\tau \hspace{-0.5mm} = \hspace{-0.5mm} \frac{c_k}{\sum_{b_k(\tau)\in g_\lambda(\tau)}c_k}p_{ij\rightarrow \lambda}^\tau \hspace{-0.5mm}\ge \hspace{-0.5mm} \frac{c_k \times \sum_{b_k(\tau)\in g_\lambda(\tau)} c_k}{\sum_{b_k(\tau)\in g_\lambda(\tau)}c_k} \hspace{-0.5mm}=\hspace{-0.5mm} c_k$.
\end{proof}

\begin{theorem}
    The CGSC costs $O(|\mathbb{T}(\tau)|^{\max}(|B(\tau)|^{\max}+\mathcal{N}\mathcal{M}^2))$, where $|\mathbb{T}(\tau)|^{\max}=  \max_{v_i\in V_{s}}|\mathbb{T}_i(\tau)|, |B(\tau)|^{\max} = \max_{v_i\in V_{s}, t_{i,j}^{\tau} \in \mathbb{T}_i(\tau)}|B_{i,j}(\tau)|$. \label{thm:complex1}
\end{theorem}
\begin{proof}
The initialization of $\overline{B}_{i,j}(\tau)$ in lines \ref{line:candi_bid_begin}-\ref{line:candi_bid_end} and the preparation of $[\overline{\mathbb{G}}_{i,j}]^1$ in line \ref{line:layer_1_candidate_group_sets_1} need $2|B_{i,j}(\tau)|$ iterations. The layer-n group construction needs $\mathcal{N}\times\frac{\mathcal{M}(\mathcal{M}-1)}{2}$ iterations in lines \ref{line:start_layer_group_constuction}-\ref{line:end_layer_group_constuction}. Finally, we need $2\times \mathcal{N}\times\frac{\mathcal{M}(\mathcal{M}-1)}{2}$ iterations to exclude the groups that do not meet constraints in terms of the traffic amount, bandwidth, and budget in lines \ref{line:traffic_bandwidth_req_check_start}-\ref{line:traffic_bandwidth_req_check_end} and \ref{line:budget_filter_begin}-\ref{line:budget_filter_end}, respectively. Based on the above analysis, the CGSC algorithm costs $O(|\mathbb{T}(\tau)|^{\max}(2|B_i(\tau)|+\mathcal{N}\times\frac{\mathcal{M}(\mathcal{M}-1)}{2}+2\times \mathcal{N}\times\frac{\mathcal{M}(\mathcal{M}-1)}{2}))=O(|\mathbb{T}(\tau)|^{\max}(|B(\tau)|^{\max}+\mathcal{N}\mathcal{M}^2))$. 
\end{proof}

\begin{theorem}
    The CSTP costs $O((|\mathbb{T}(\tau)|^{\max})|\mathbb{G}(\tau)|^{\max})$, where $|\mathbb{G}(\tau)|^{\max}=\max_{v_i\in V_{s}, t_{i,j}^\tau \in \mathbb{T}_i(\tau)}|\mathbb{G}_{i,j}(\tau)|$ and $|\mathbb{T}(\tau)|^{\max}=  \max_{v_i\in V_{s}}|\mathbb{T}_i(\tau)|$. \label{thm4} 
    \label{thm3_1}
\end{theorem}
\begin{proof}
For each task, we need to iterate the set $\mathbb{G}_{i,j}(\tau)$ to find the collaborator group. In each iteration, we need to iterate the set $\mathbb{G}_{i,j}(\tau)$ twice to get the candidate group and payment in lines \ref{line:co_end} and \ref{line:u_define}. Thus, the CSTP costs $O(|\mathbb{T}(\tau)|^{\max}(|\mathbb{G}(\tau)|^{\max}+|\mathbb{G}(\tau)|^{\max}))=O((|\mathbb{T}(\tau)|^{\max})|\mathbb{G}(\tau)|^{\max})$. 
\end{proof}

\begin{theorem}
    The regret of the CSTP, which is the difference between the utility given by the CSTP and optimal value, is $\max_{t_{i,j}^\tau \in \mathbb{T}_i(\tau), g_{\lambda}(\tau) \in \mathbb{G}_{i,j}(\tau)}{\sum_{b_k(\tau)\in g_{\lambda}(\tau)}c_k}\overline{R}$, where we have $\overline{R}=8|\mathbb{T}_i(\tau)|\sum_{\lambda:U_{\max}^{\tau}  > U_{i,j\rightarrow \lambda}^\tau }\frac{\ln n_{sum}}{\Delta_{\lambda}} + (1+\frac{\pi^2}{3})\sum_{g_{\lambda}(\tau)\in \mathbb{G}_{i,j}(\tau)}\Delta_{\lambda}$, and then we have $U_{\max}^{\tau}=\max_{t_{i,j}^\tau \in \mathbb{T}_i(\tau), g_{\lambda}(\tau) \in \mathbb{G}_{i,j}(\tau)}\frac{U_{i,j\rightarrow \lambda}^\tau}{\sum_{b_k(\tau)\in g_{\lambda}(\tau)}c_k}$ and $\Delta_{\lambda}=U_{\max}^{\tau}-U_{i,j\rightarrow \lambda}^\tau$.
\end{theorem}
\begin{proof}
The CSTP, which incorporates an exploration factor to provide the groups with fewer number of times being selected more chance to be selected, can be viewed as an extension of the classic Upper Confidence Bound algorithm. Based on \cite{auer2002finite}, we know that the regret is upper bounded by $\overline{R}=8|\mathbb{T}_i(\tau)|\sum_{\lambda:U_{\max}^{\tau}  > U_{i,j\rightarrow \lambda}^\tau }\frac{\ln n_{sum}}{\Delta_{\lambda}} + (1+\frac{\pi^2}{3})\sum_{g_{\lambda}(\tau)\in \mathbb{G}_{i,j}(\tau)}\Delta_{\lambda}$, where we have $U_{\max}^{\tau}=\max_{t_{i,j}^\tau \in \mathbb{T}_i(\tau), g_{\lambda}(\tau) \in \mathbb{G}_{i,j}(\tau)}\frac{U_{i,j\rightarrow \lambda}^\tau}{\sum_{b_k(\tau)\in g_{\lambda}(\tau)}c_k}$ and $\Delta_{\lambda}=U_{\max}^{\tau}-U_{i,j\rightarrow \lambda}^\tau$. Given that the CSTP considers the ratio between the total utility and the total cost of each collaborator group during auction, the regret of the CSTP is $\max_{t_{i,j}^\tau \in \mathbb{T}_i(\tau), g_{\lambda}(\tau) \in \mathbb{G}_{i,j}(\tau)}{\sum_{b_k(\tau)\in g_{\lambda}(\tau)}c_k}\overline{R}$.
\end{proof}

%% file: sections/evaluation.tex
\section{Performance Evaluation}
\label{sec:eval}

\begin{figure*}[!htb]
  \centering
  \includegraphics[width=\textwidth]{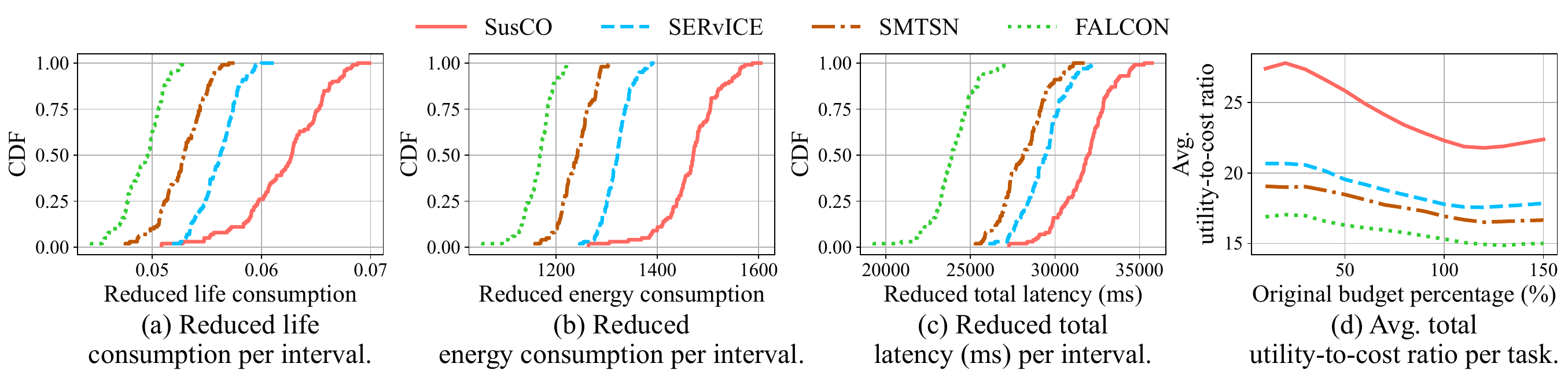}
  \caption{Performance metric comparison with other schemes.}
  \label{fig:cdf_all_metrics}
\end{figure*}

\subsection{Simulation Settings}

\textbf{LEO satellite network parameters.} We use the first-shell parameters of the largest real-world-based LEO satellite network, Starlink \cite{giuliari2021icarus}, to evaluate our proposed \ourFrameworkName{} framework. Furthermore, we test the performance and runtime of our proposed algorithms in additional real-world-based LEO satellite networks, Telesat \cite{DELPORTILLO2019123}, OneWeb \cite{DELPORTILLO2019123}, and Kuiper \cite{kuiper_fcc_report_2021}, to illustrate the robustness and scalability of our algorithms. The aggregated amount of data received by the source LEO satellite is set as $300$ Mb/s according to \cite{wang_enhancing_eo_2022}.

\textbf{Commercial dish parameters.}  The locations of the ground stations and 5G New Radio base stations are based on the locations of AWS Ground Stations \cite{aws_ground_station_locations_2024} and Internet user distribution used in \cite{yang_towards_2016}, respectively. The bandwidth of the ground stations and 5G New Radio base stations is set based on \cite{aws_ground_station_faqs_2024} and \cite{deng_how_much_to_pay_2020}, respectively.

\textbf{Other parameters.} The unit price $\alpha$ and $\beta$ for data amount and bandwidth reservation time are $0.09$ / GB and $0.17$ / second according to Amazon's EC2 pricing \cite{aws_ec2_pricing_on_demand_2024} and Microsoft Azure Orbital Ground Station pricing \cite{azure_orbital_pricing_2024}, respectively. The length of each interval is set as $60$ seconds \cite{chen_mobility_2021} and the simulations have $100$ intervals. We assume each LEO satellite has the same energy consumption rate $\epsilon_m$ and set it as $0.08$ Watt / Mbps according to \cite{yang_towards_2016}. The parameters for the battery life consumption are set based on \cite{yang_towards_2016}. The weights $\omega_1$, $\omega_2$, and $\omega_3$ are set as $0.3$, $0.4$, and $0.3$, respectively, where $\omega_2$ is higher since latency is a significant QoS. We set $\mathcal{N}$ as $2$ to reduce group size as only a limited number of GSLs can be established. We set $\mathcal{M}$ as $10$ to reduce the algorithm runtime by controlling the number of possible collaborator groups.

\textbf{Comparison settings.} We compare our \ourFrameworkName{} framework with three state-of-the-art schemes for LEO satellite network data offloading:

\begin{enumerate}
    \item  SERvICE scheme \cite{li_service_2018} calculates the normalized utility based on the estimated latency and bandwidth with the same weight for routing the data to the destination.

    \item  SMTSN scheme \cite{chou_sustainable_2022} calculates the normalized utility of sustainability and latency improvement with the same weight for each dish, using the extended service life and the estimated latency for routing the data to the destination, respectively.

    \item FALCON scheme \cite{lyu_falcon_2023} examines the offloading paths proposed by the dishes and selects the one with the lowest latency to minimize the total flow completion time.

\end{enumerate}

\subsection{Evaluation Results}

\begin{figure*}[!htb]
  \centering
  \includegraphics[width=\textwidth]{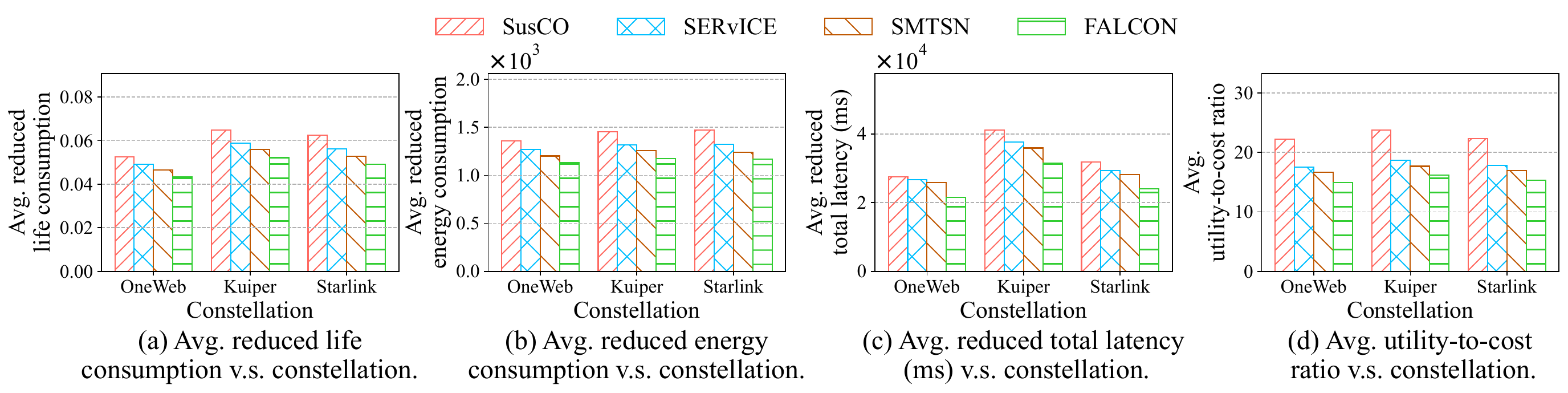}
  \caption{Performance metric comparison in different constellations.}
  \label{fig:barcharts_all_metrics}
\end{figure*}

\begin{figure}[t]
    \centering
    \begin{subfigure}[t]{0.49\linewidth}
        \centering
        \includegraphics[width=\linewidth]{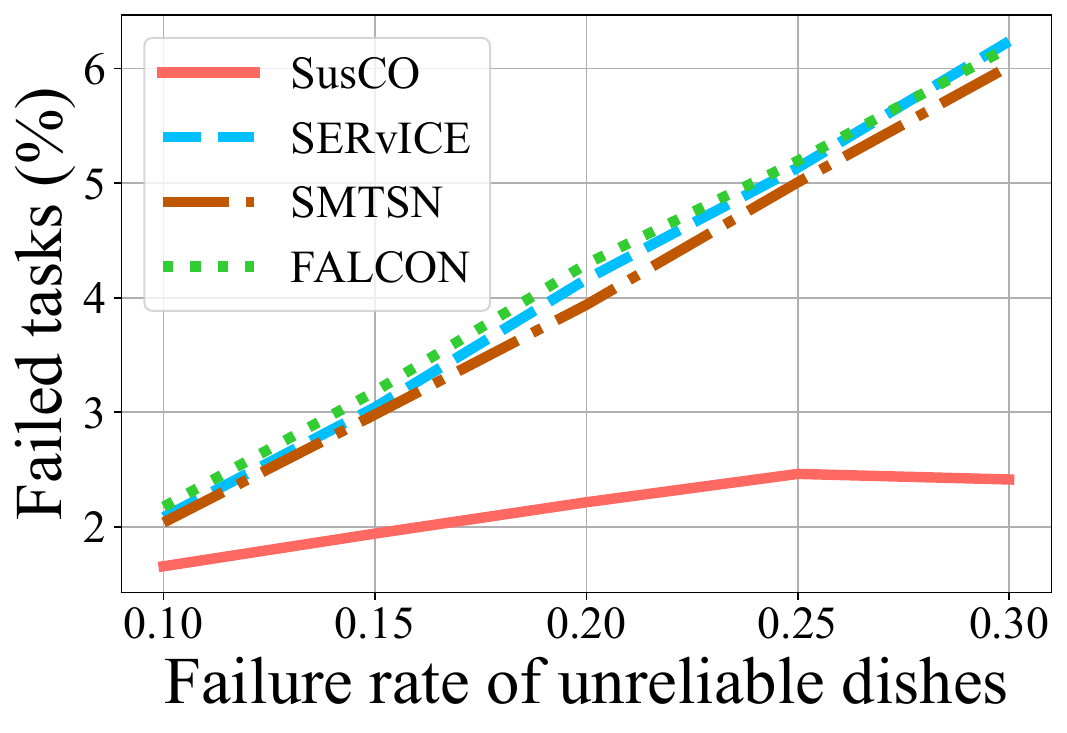} 
        \caption{Failure rate robustness.}
        \label{fig:failure_vs_failure_rate_of_unreliable_dishes}
    \end{subfigure}
    \hfill
    \centering
    \begin{subfigure}[t]{0.49\linewidth}
        \centering
        \includegraphics[width=\linewidth]{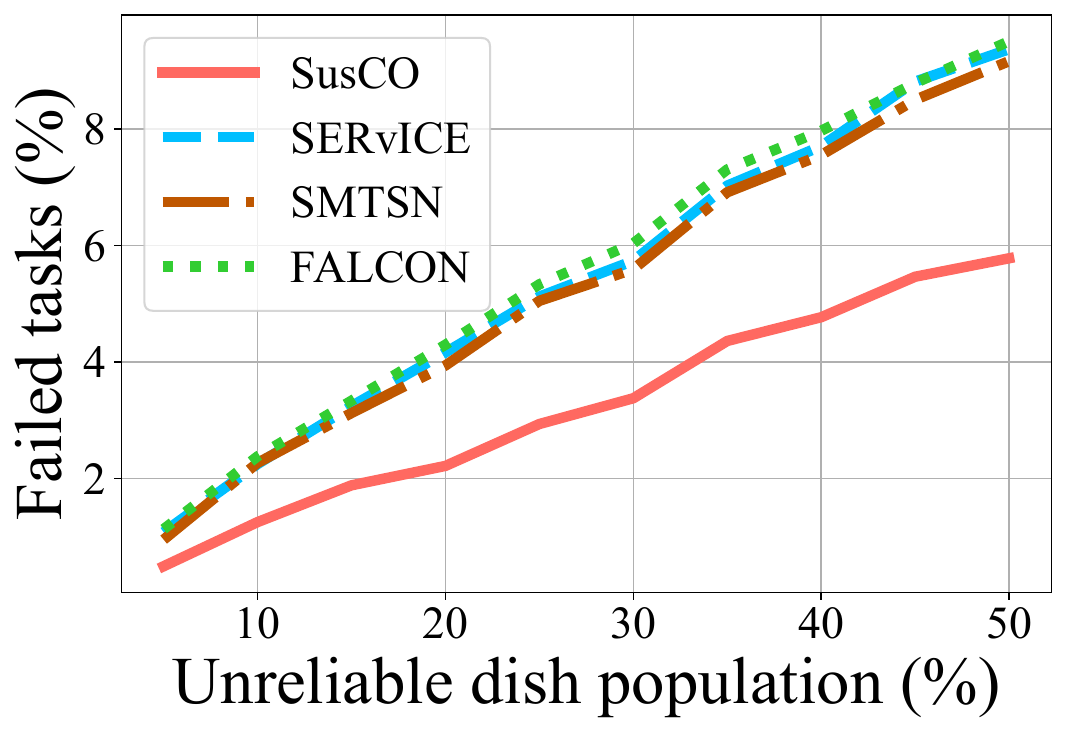} 
        \caption{Population robustness.}
        \label{fig:failure_vs_unreliable_dish_proportion}
    \end{subfigure}
    \caption{Performance comparison for offloading robustness.}
    \vspace{-3mm}
\end{figure}

\begin{figure}[t]
    \mbox{
        \begin{minipage}[t]{0.49\linewidth}
            \centering
            \includegraphics[width=\linewidth]{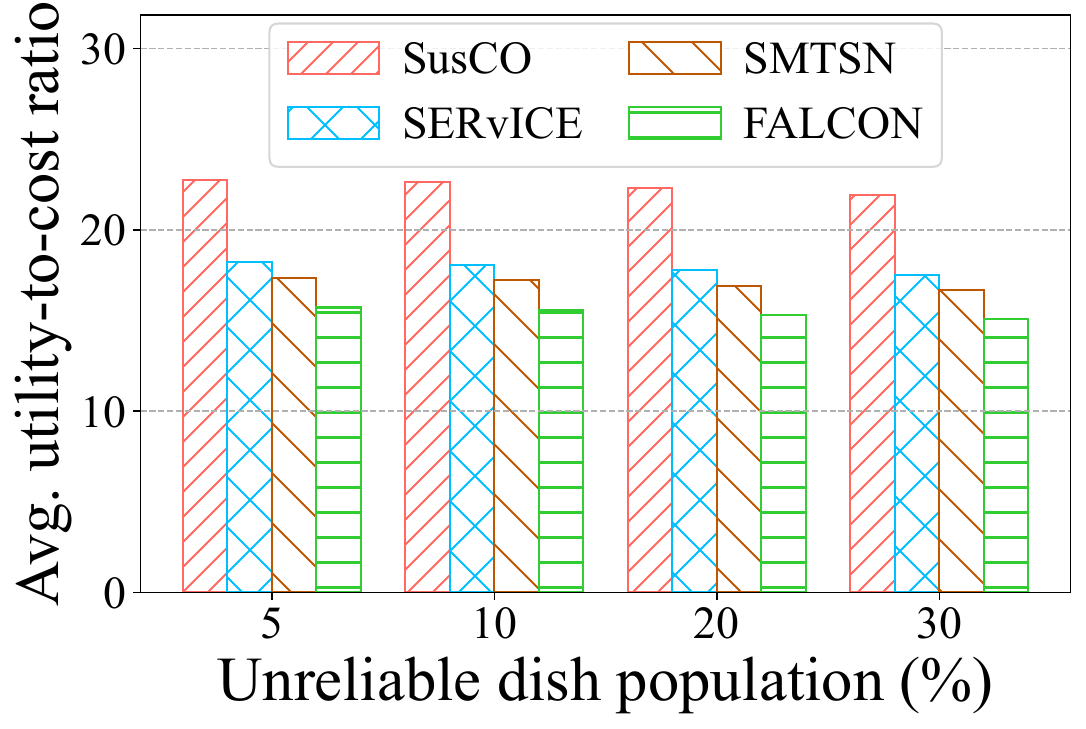} 
            \caption{Average utility-to-cost ratio v.s. unreliable dishes.} \label{fig:avg_uc_vs_unreliable_dish_population}
        \end{minipage}

        \begin{minipage}[t]{0.49\linewidth}
            \centering
            \includegraphics[width=\linewidth]{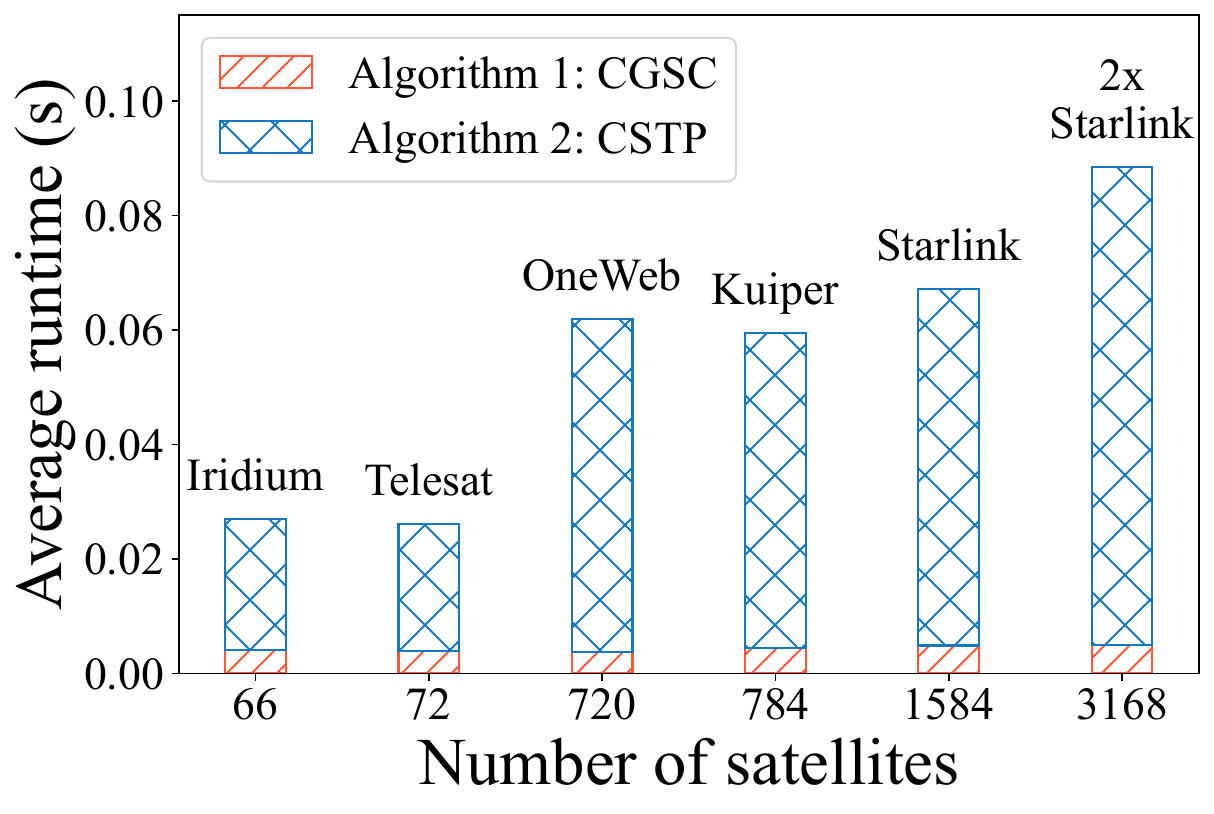} 
            \caption{Algorithm runtime.} \label{fig:algo_avg_runtime}
        \end{minipage}
    } 
\end{figure}

\subsubsection{Sustainable Extended Service Life}

Fig. \ref{fig:cdf_all_metrics}a shows the CDF of the reduced life consumption for the LEO satellite batteries per interval. Our \ourFrameworkName{} reduces $\lifeConsumptionReductionPercentageLow\%$ to $\lifeConsumptionReductionPercentageHigh\%$ more life consumption on average per interval compared to the other schemes since it can offload more tasks to the dishes in terrestrial networks with a limited budget.

By offloading tasks to terrestrial networks, the battery life consumption for the LEO satellite hops on the path after the offloading satellite can be reduced. This can significantly prolong the lifespan of those satellite batteries and service life.

\subsubsection{Efficient Energy Consumption}

As shown in Fig \ref{fig:cdf_all_metrics}b, our \ourFrameworkName{} reduces $\energyReductionPercentageLow\%$ to $\energyReductionPercentageHigh\%$ more energy consumption than the other schemes due to its ability to offload more tasks to terrestrial networks. The payment scheme of \ourFrameworkName{} effectively manages the limited budget to offload more tasks. 

After the tasks are offloaded, the energy consumption of the LEO satellite hops after the offloading satellite on the path is reduced. Furthermore, reducing energy consumption also reduces the negative impact on the battery lifespan for those satellites with low battery levels.

\subsubsection{Improved End-to-End Latency}

Fig. \ref{fig:cdf_all_metrics}c shows the CDF of the reduced total latency (ms) for all tasks per interval. Our \ourFrameworkName{} reduces $\delayReductionPercentageLow\%$ to $\delayReductionPercentageHigh\%$ more end-to-end latency per interval compared to the other schemes as \ourFrameworkName{} can effectively manage the limited budget and offload more tasks. 

The offloaded tasks can avoid the congested LEO satellite network and achieve better latency since \ourFrameworkName{} considers the estimated latency for routing data to the destination when calculating the utility. Furthermore, offloading tasks can also mitigate the congestion in the LEO satellite network.

\subsubsection{Efficient Use of Budget}

Fig. \ref{fig:cdf_all_metrics}d shows the average total utility-to-cost ratio per task in different amounts of budget. Our \ourFrameworkName{} has the highest average utility-to-cost ratio compared to the other schemes. A higher utility-to-cost ratio indicates more effective budget management for offloading. \ourFrameworkName{} uses the budget more effectively than the other schemes since it focuses on the utility-to-cost ratio, maximizing the total utility purchased by the unit budget.

While the budget is effectively managed for the LEO satellite network operator, the payments determined by \ourFrameworkName{} also provide dish providers with incentives to participate in auctions and share their resources.

\subsubsection{Robustness and Unreliable Dishes}
We evaluate the robustness of our \ourFrameworkName{} when there are unreliable dishes participating in auctions. The unreliable dishes are defined as those dishes that have a failure rate significantly larger than the industry average.

\textbf{Robust on high failure rates.} Fig. \ref{fig:failure_vs_failure_rate_of_unreliable_dishes} illustrates the percentages of failed offloading tasks with respect to the failure rate of the unreliable dishes. Our \ourFrameworkName{} framework maintains an almost log-scale and the lowest percentage of failed offloading tasks because it considers the failure rate in the discounted total utility as shown in Eq. (\ref{eq:discounted_utility}), where the dishes with high failure rates are less likely to be selected by our \ourFrameworkName.

\textbf{Robust on large unreliable dish population.} Fig. \ref{fig:failure_vs_unreliable_dish_proportion} also illustrates the percentage of failed offloading tasks but with respect to the population percentage of the unreliable dishes. Since any dish providers in terrestrial networks can bid, there may be a time when a significant proportion of dishes are unreliable. Our \ourFrameworkName{} demonstrates its ability to select reliable dishes even if the number of unreliable dishes increases, having the lowest percentage of failed offloading tasks compared to the other schemes.

\subsubsection{Robust Performance} 

We measure the robustness of the average performance in different constellations, as shown in Fig. \ref{fig:barcharts_all_metrics}. Our \ourFrameworkName{} is able to achieve the highest average performance in different constellations with different numbers of orbits, numbers of satellites per orbit, and inclination angles. 

Different constellations have unique network dynamics and ground coverage patterns. Our \ourFrameworkName{} can seamlessly adapt to the unique characteristics of each LEO satellite network, demonstrating its ability to be applied universally.

We measure the robustness of the average utility-to-cost ratios in different unreliable dish populations, as shown in Fig. \ref{fig:avg_uc_vs_unreliable_dish_population}. Our \ourFrameworkName{} is able to achieve the highest average utility-to-cost ratios even with larger unreliable dish populations.

\ourFrameworkName{} discounts the utility by considering the past failure history of the bidders, making the bidders with abnormal failure rates much less likely to win the task while also promoting the bidders with extraordinarily low failure rates.

\subsubsection{Algorithm Scalability}

\begin{table}[t]
\centering
\begin{tabular}{|c|c|c|c|c|}
\hline
\diagbox[width=3cm, height=1cm]{Runtime (ms)}{Size}{Constellation} & Telesat & Kuiper & Starlink & 2xStarlink \\
\hline
Number of satellites & 72 & 784 & 1584 & 2x1584 \\
\hline
CGSC & 6.31 & 5.07 & 5.50 & 5.32 \\
\hline
CSTP & 35.22 & 61.75 & 69.61 & 90.02 \\
\hline
Total & 41.53 & 66.82 & 75.11 & 95.34 \\
\hline
\end{tabular}
\caption{Algorithm runtime in different constellations.}
\label{tab:runtime_comparison}
\vspace{-6mm}
\end{table}

We record the runtime of our proposed algorithms on the LEO satellite networks in different constellations and other real-world-based parameters using an AMD EPYC 7313 3.0 GHz CPU as shown in Fig. \ref{fig:algo_avg_runtime} and TABLE \ref{tab:runtime_comparison}. 

\textbf{Constant runtime to construct candidate sets.} The CGSC algorithm takes a near-constant runtime on larger constellations to construct a candidate collaborator group set because we control the size of the candidate collaborator group set for each task, making it scalable on larger constellations. 

\textbf{Linear runtime to select winners.} The CSTP algorithm takes a linear runtime to select winners and determine payments for all tasks since it takes the collaborator group set from the CGSC algorithm to select the winners, where the size of the collaborator group set is controlled by the parameters.

%% file: sections/related_work.tex
\section{Related Work}
\label{sec:related_work}

\textbf{Bent-pipe routing.} The current largest LEO satellite network, Starlink, is mainly using the bent-pipe routing strategy as revealed by the measurement study conducted by Ma et al. \cite{ma_starlink_measure_2023}. As bent-pipe routing strategy still remains an important role in the industry, many early works \cite{hauri_internet_from_space_2020, Handley_ground_relays_2019, kassing_hypatia_2020, pan_relay_2023} studied the bent-pipe routing strategy. Furthermore, integrated space and terrestrial networks (ISTNs) can be considered as an extension of the bent-pipe routing strategy \cite{li_service_2018, chen_time_varying_2021, zhang_enabling_low_latency_2022, li_stable_hier_2024}. For example, Li et al. \cite{li_service_2018} proposed a QoS-based framework that assigns different weights to different QoS aspects in ISTNs.

Different from their works, we focus on a more practical scenario where dish providers in terrestrial networks require incentives, such as compensation for their costs, to cooperate with LEO satellite operators for data offloading.

\textbf{Inter-satellite link routing.} Inter-satellite links (ISLs) are considered a promising solution to realize low-latency routing in LEO satellite networks \cite{Handley_delay_not_option_2018}. Existing studies \cite{yang_towards_2016, chou_sustainable_2022, chou_fcsn_2024} focused on the in-space routing using ISLs to improve sustainability. For example, Yang et al. \cite{yang_towards_2016} proposed a routing algorithm based on satellite battery levels to extend satellite battery lifespan. Other existing work \cite{huang_Pheromone_2022} focused on utilizing ISLs to improve application QoS.

In addition to the routing using ISLs, we also focus on mitigating the sustainability issues using ground-satellite links (GSLs) to alleviate the negative impact of the traffic load on the battery lifespan.

\textbf{Ground station as a service routing.} As the market of Ground Station as a Service emerges, some existing studies \cite{singh_community_2021, lai_orbitcast_2021, lyu_falcon_2023} focused on routing data by utilizing commercial or shared ground stations. For example, Lai et al. \cite{lai_orbitcast_2021} exploited Ground Station as a Service to receive data from Earth observation satellites and Lyu et al. \cite{lyu_falcon_2023} focused on multi-path routing where Ground Station as a Service can be one of the paths. Singh et al. \cite{singh_community_2021} proposed a framework for satellite community users to share their ground station resources.

These studies complement our work. While they also focus on commercial or shared ground stations, we focus on building a framework where any ground station providers can participate and receive payment as compensation, while considering their service QoS.

%% file: sections/further_discussion.tex
\section{Future Work}
\label{sec:future_work}

Since all commercial dish providers can participate in our \ourFrameworkName{} framework, it is also possible that those unreliable providers with poor reputations for fulfilling the contract can bid for the offloading tasks. If one provider in the collaborator group fails to receive the offloading data, the offloading task will fail despite other reliable providers in the group successfully receiving the offloading data.

One intuitive way is to maintain a manual list of reliable and unreliable providers. However, this method is inefficient and may delay the onboarding of new reliable providers. New entrants may need to apply and wait for approval before they can participate in any offloading tasks. This can potentially slow down the market growth.

To ensure the market is efficient within our \ourFrameworkName{} framework, it is crucial to have an automated system for tracking provider performance. The objective of this system is to provide rewards for those reliable providers who fulfill the contract and penalize those unreliable providers who fail to do so. Our future work is to develop this automated system and integrate it with our \ourFrameworkName{} framework.

%% file: sections/conclusion.tex
\section{Conclusion}
\label{sec:conclusion}

In this paper, we proposed a sustainable and cost-effective offloading framework named Sustainable Collaborative Offloading (\ourFrameworkName) for LEO satellite networks to offload data to commercial dishes. We formulated the Collaborative Dishes as Ground Stations problem and proposed the collaborator group set construction algorithm to construct candidate collaborator groups and the collaborator selection and total payment algorithm to select winners and determine payments. Finally, we ran extensive simulations and the results showed that our solution significantly improved energy consumption, satellite service life, and end-to-end latency.

%% file: sections/ack.tex
\section*{Acknowledgements}
\label{acknowledgements}

This research was generously supported by an NSERC Discovery Grant.

%% file: main.bbl
\begin{thebibliography}{10}
\providecommand{\url}[1]{#1}
\csname url@samestyle\endcsname
\providecommand{\newblock}{\relax}
\providecommand{\bibinfo}[2]{#2}
\providecommand{\BIBentrySTDinterwordspacing}{\spaceskip=0pt\relax}
\providecommand{\BIBentryALTinterwordstretchfactor}{4}
\providecommand{\BIBentryALTinterwordspacing}{\spaceskip=\fontdimen2\font plus
\BIBentryALTinterwordstretchfactor\fontdimen3\font minus \fontdimen4\font\relax}
\providecommand{\BIBforeignlanguage}[2]{{%
\expandafter\ifx\csname l@#1\endcsname\relax
\typeout{** WARNING: IEEEtran.bst: No hyphenation pattern has been}%
\typeout{** loaded for the language `#1'. Using the pattern for}%
\typeout{** the default language instead.}%
\else
\language=\csname l@#1\endcsname
\fi
#2}}
\providecommand{\BIBdecl}{\relax}
\BIBdecl

\bibitem{su_broadband_2019}
Y.~Su, Y.~Liu, Y.~Zhou, J.~Yuan, H.~Cao, and J.~Shi, ``Broadband leo satellite communications: Architectures and key technologies,'' \emph{IEEE Wireless Communications}, vol.~26, no.~2, pp. 55--61, 2019.

\bibitem{xiao_leo_6g_2024}
Z.~Xiao, J.~Yang, T.~Mao, C.~Xu, R.~Zhang, Z.~Han, and X.-G. Xia, ``Leo satellite access network (leo-san) toward 6g: Challenges and approaches,'' \emph{IEEE Wireless Communications}, vol.~31, no.~2, pp. 89--96, 2024.

\bibitem{Giuliari_backbones_space_2020}
G.~Giuliari, T.~Klenze, M.~Legner, D.~Basin, A.~Perrig, and A.~Singla, ``Internet backbones in space,'' \emph{SIGCOMM Comput. Commun. Rev.}, vol.~50, no.~1, p. 25–37, mar 2020.

\bibitem{ma_starlink_measure_2023}
S.~Ma, Y.~C. Chou, H.~Zhao, L.~Chen, X.~Ma, and J.~Liu, ``Network characteristics of leo satellite constellations: A starlink-based measurement from end users,'' in \emph{Proc. IEEE INFOCOM}, 2023.

\bibitem{singh_community_2021}
V.~Singh, A.~Prabhakara, D.~Zhang, O.~Ya\u{g}an, and S.~Kumar, ``A community-driven approach to democratize access to satellite ground stations,'' in \emph{Proc. ACM MobiCom}, 2021.

\bibitem{vas_l2d2_2021}
D.~Vasisht, J.~Shenoy, and R.~Chandra, ``L2d2: Low latency distributed downlink for leo satellites,'' in \emph{Proc. ACM SIGCOMM}, 2021.

\bibitem{chou_sustainable_2022}
Y.~C. Chou, X.~Ma, F.~Wang, S.~Ma, S.~H. Wong, and J.~Liu, ``Towards sustainable multi-tier space networking for leo satellite constellations,'' in \emph{Proc. IEEE/ACM IWQoS}, 2022.

\bibitem{aws_gs_sla_2022}
``Amazon gs sla,'' https://aws.amazon.com/ground-station/sla/.

\bibitem{azure_orbital_sla_2022}
``Azure orbital sla,'' https://azure.microsoft.com/en-ca/products/orbital.

\bibitem{3GPP_TR_38811}
{3rd Generation Partnership Project (3GPP)}, ``{Study on New Radio (NR) to support non-terrestrial networks},'' 3GPP, Technical Report, 2020, version 15.4.0.

\bibitem{yang_towards_2016}
Y.~Yang, M.~Xu, D.~Wang, and Y.~Wang, ``Towards energy-efficient routing in satellite networks,'' \emph{IEEE Journal on Selected Areas in Communications}, vol.~34, no.~12, pp. 3869--3886, 2016.

\bibitem{li_leotp_2023}
L.~Jiang, Y.~Zhang, J.~Yin, X.~Zhang, and B.~Liu, ``Leotp: An information-centric transport layer protocol for leo satellite networks,'' in \emph{Proc. IEEE ICDCS}, 2023.

\bibitem{chou_fcsn_2024}
Y.~C. Chou, L.~Chen, F.~Wang, H.~Wang, X.~Ma, S.~Ma, and J.~Liu, ``Orchestrating sustainable and service-differentiable satellite networking: A federated cross-orbit approach,'' in \emph{Proc. IEEE/ACM IWQoS}, 2024.

\bibitem{huang_Pheromone_2022}
Y.~Huang, X.~Jiang, S.~Chen, F.~Yang, and J.~Yang, ``Pheromone incentivized intelligent multipath traffic scheduling approach for leo satellite networks,'' \emph{IEEE Transactions on Wireless Communications}, vol.~21, no.~8, pp. 5889--5902, 2022.

\bibitem{li_service_2018}
T.~Li, H.~Zhou, H.~Luo, and S.~Yu, ``Service: A software defined framework for integrated space-terrestrial satellite communication,'' \emph{IEEE Transactions on Mobile Computing}, vol.~17, no.~3, pp. 703--716, 2018.

\bibitem{chen_time_varying_2021}
L.~Chen, F.~Tang, Z.~Li, L.~T. Yang, J.~Yu, and B.~Yao, ``Time-varying resource graph based resource model for space-terrestrial integrated networks,'' in \emph{Proc. IEEE INFOCOM}, 2021.

\bibitem{zhang_enabling_low_latency_2022}
Y.~Zhang, Q.~Wu, Z.~Lai, and H.~Li, ``Enabling low-latency-capable satellite-ground topology for emerging leo satellite networks,'' in \emph{Proc. IEEE INFOCOM}, 2022.

\bibitem{li_stable_hier_2024}
Y.~Li, L.~Liu, H.~Li, W.~Liu, Y.~Chen, W.~Zhao, J.~Wu, Q.~Wu, J.~Liu, and Z.~Lai, ``Stable hierarchical routing for operational leo networks,'' in \emph{Proc. ACM MobiCom}, 2024.

\bibitem{deng_how_many_leo_2021}
R.~Deng, B.~Di, H.~Zhang, L.~Kuang, and L.~Song, ``Ultra-dense leo satellite constellations: How many leo satellites do we need?'' \emph{IEEE Transactions on Wireless Communications}, vol.~20, no.~8, pp. 4843--4857, 2021.

\bibitem{NASA2024_commercial_gs}
``{NASA}’s commercial communications services,'' https://www.nasa.gov/directorates/somd/space-communications-navigation-program/nasas-commercial-communications-services/.

\bibitem{zdnet_2023_spacex}
``Spacex: We've launched 32,000 linux computers into space for starlink internet,'' https://www.zdnet.com/article/spacex-weve-launched-32000-linux-computers-into-space-for-starlink-internet/.

\bibitem{myerson1981optimal}
R.~B. Myerson, ``Optimal auction design,'' \emph{Mathematics of Operations Research}, vol.~6, no.~1, pp. 58--73, 1981.

\bibitem{auer2002finite}
P.~Auer, N.~Cesa-Bianchi, and P.~Fischer, ``Finite-time analysis of the multiarmed bandit problem,'' \emph{Machine learning}, vol.~47, pp. 235--256, 2002.

\bibitem{giuliari2021icarus}
G.~Giuliari, T.~Ciussani, A.~Perrig, and A.~Singla, ``Icarus: Attacking low earth orbit satellite networks.'' in \emph{Proc. USENIX ATC}, 2021.

\bibitem{DELPORTILLO2019123}
I.~{del Portillo}, B.~G. Cameron, and E.~F. Crawley, ``A technical comparison of three low earth orbit satellite constellation systems to provide global broadband,'' \emph{Acta Astronautica}, vol. 159, pp. 123--135, 2019.

\bibitem{kuiper_fcc_report_2021}
{Federal Communications Commission}, ``Application for fixed satellite service by kuiper systems llc,'' \url{https://fcc.report/IBFS/SAT-LOA-20211104-00145}, 2021.

\bibitem{wang_enhancing_eo_2022}
P.~Wang, H.~Li, B.~Chen, and S.~Zhang, ``Enhancing earth observation throughput using inter-satellite communication,'' \emph{IEEE Transactions on Wireless Communications}, vol.~21, no.~10, pp. 7990--8006, 2022.

\bibitem{aws_ground_station_locations_2024}
``Aws ground station locations,'' https://aws.amazon.com/ground-station/locations/.

\bibitem{aws_ground_station_faqs_2024}
``Aws ground station faqs,'' https://aws.amazon.com/ground-station/faqs/.

\bibitem{deng_how_much_to_pay_2020}
R.~Deng, B.~Di, S.~Chen, S.~Sun, and L.~Song, ``Ultra-dense leo satellite offloading for terrestrial networks: How much to pay the satellite operator?'' \emph{IEEE Transactions on Wireless Communications}, vol.~19, no.~10, pp. 6240--6254, 2020.

\bibitem{aws_ec2_pricing_on_demand_2024}
``{AWS EC2} pricing,'' https://aws.amazon.com/ec2/pricing/on-demand/.

\bibitem{azure_orbital_pricing_2024}
``Azure orbital earth station – pricing,'' https://azure.microsoft.com/en-ca/pricing/details/orbital/.

\bibitem{chen_mobility_2021}
L.~Chen, F.~Tang, and X.~Li, ``Mobility- and load-adaptive controller placement and assignment in leo satellite networks,'' in \emph{Proc. IEEE INFOCOM}, 2021.

\bibitem{lyu_falcon_2023}
M.~Lyu, Q.~Wu, Z.~Lai, H.~Li, Y.~Li, and J.~Liu, ``Falcon: Towards fast and scalable data delivery for emerging earth observation constellations,'' in \emph{Proc. IEEE INFOCOM}, 2023.

\bibitem{hauri_internet_from_space_2020}
Y.~Hauri, D.~Bhattacherjee, M.~Grossmann, and A.~Singla, ``"internet from space" without inter-satellite links,'' in \emph{Proc. ACM HotNets}, 2020.

\bibitem{Handley_ground_relays_2019}
M.~Handley, ``Using ground relays for low-latency wide-area routing in megaconstellations,'' in \emph{Proc. ACM HotNets}, 2019.

\bibitem{kassing_hypatia_2020}
S.~Kassing, D.~Bhattacherjee, A.~B. \'{A}guas, J.~E. Saethre, and A.~Singla, ``Exploring the "internet from space" with hypatia,'' in \emph{Proc. ACM IMC}, 2020.

\bibitem{pan_relay_2023}
G.~Pan, J.~Ye, J.~An, and M.-S. Alouini, ``Latency versus reliability in leo mega-constellations: Terrestrial, aerial, or space relay?'' \emph{IEEE Transactions on Mobile Computing}, vol.~22, no.~9, pp. 5330--5345, 2023.

\bibitem{Handley_delay_not_option_2018}
M.~Handley, ``Delay is not an option: Low latency routing in space,'' in \emph{Proc. ACM HotNets}, 2018.

\bibitem{lai_orbitcast_2021}
Z.~Lai, Q.~Wu, H.~Li, M.~Lv, and J.~Wu, ``Orbitcast: Exploiting mega-constellations for low-latency earth observation,'' in \emph{Proc. IEEE ICNP}, 2021.

\end{thebibliography}
